\pdfoutput=1

\documentclass{amsart}

\usepackage{ownstyle}

\hypersetup{
pagebackref=true
, backref=page
}

\begin{document}


\makeatletter
\newcommand{\sectionNotes}{\phantomsection\section*{Notes}\addcontentsline{toc}{section}{Notes}\markright{\textsc{\@chapapp{} \thechapter{} Notes}}}
\newcommand{\sectionExercises}[1]{\phantomsection\section*{Exercises}\addcontentsline{toc}{section}{Exercises}\markright{\textsc{\@chapapp{} \thechapter{} Exercises}}}
\makeatother

\newcommand{\jdeq}{\equiv}      
\let\judgeq\jdeq
\newcommand{\defeq}{\vcentcolon\equiv}  

\newcommand{\define}[1]{\textbf{#1}}


\newcommand{\Vect}{\ensuremath{\mathsf{Vec}}}
\newcommand{\Fin}{\ensuremath{\mathsf{Fin}}}
\newcommand{\fmax}{\ensuremath{\mathsf{fmax}}}
\newcommand{\seq}[1]{\langle #1\rangle}

\def\prdsym{\textstyle\prod}
\makeatletter
\def\prd#1{\@ifnextchar\bgroup{\prd@parens{#1}}{%
    \@ifnextchar\sm{\prd@parens{#1}\@eatsm}{%
    \@ifnextchar\prd{\prd@parens{#1}\@eatprd}{%
    \@ifnextchar\;{\prd@parens{#1}\@eatsemicolonspace}{%
    \@ifnextchar\\{\prd@parens{#1}\@eatlinebreak}{%
    \@ifnextchar\narrowbreak{\prd@parens{#1}\@eatnarrowbreak}{%
      \prd@noparens{#1}}}}}}}}
\def\prd@parens#1{\@ifnextchar\bgroup%
  {\mathchoice{\@dprd{#1}}{\@tprd{#1}}{\@tprd{#1}}{\@tprd{#1}}\prd@parens}%
  {\@ifnextchar\sm%
    {\mathchoice{\@dprd{#1}}{\@tprd{#1}}{\@tprd{#1}}{\@tprd{#1}}\@eatsm}%
    {\mathchoice{\@dprd{#1}}{\@tprd{#1}}{\@tprd{#1}}{\@tprd{#1}}}}}
\def\@eatsm\sm{\sm@parens}
\def\prd@noparens#1{\mathchoice{\@dprd@noparens{#1}}{\@tprd{#1}}{\@tprd{#1}}{\@tprd{#1}}}
\def\lprd#1{\@ifnextchar\bgroup{\@lprd{#1}\lprd}{\@@lprd{#1}}}
\def\@lprd#1{\mathchoice{{\textstyle\prod}}{\prod}{\prod}{\prod}({\textstyle #1})\;}
\def\@@lprd#1{\mathchoice{{\textstyle\prod}}{\prod}{\prod}{\prod}({\textstyle #1}),\ }
\def\tprd#1{\@tprd{#1}\@ifnextchar\bgroup{\tprd}{}}
\def\@tprd#1{\mathchoice{{\textstyle\prod_{(#1)}}}{\prod_{(#1)}}{\prod_{(#1)}}{\prod_{(#1)}}}
\def\dprd#1{\@dprd{#1}\@ifnextchar\bgroup{\dprd}{}}
\def\@dprd#1{\prod_{(#1)}\,}
\def\@dprd@noparens#1{\prod_{#1}\,}

\def\@eatnarrowbreak\narrowbreak{%
  \@ifnextchar\prd{\narrowbreak\@eatprd}{%
    \@ifnextchar\sm{\narrowbreak\@eatsm}{%
      \narrowbreak}}}
\def\@eatlinebreak\\{%
  \@ifnextchar\prd{\\\@eatprd}{%
    \@ifnextchar\sm{\\\@eatsm}{%
      \\}}}
\def\@eatsemicolonspace\;{%
  \@ifnextchar\prd{\;\@eatprd}{%
    \@ifnextchar\sm{\;\@eatsm}{%
      \;}}}

\def\lam#1{{\lambda}\@lamarg#1:\@endlamarg\@ifnextchar\bgroup{.\,\lam}{.\,}}
\def\@lamarg#1:#2\@endlamarg{\if\relax\detokenize{#2}\relax #1\else\@lamvar{\@lameatcolon#2},#1\@endlamvar\fi}
\def\@lamvar#1,#2\@endlamvar{(#2\,{:}\,#1)}
\def\@lameatcolon#1:{#1}
\let\lamt\lam
\def\lamu#1{{\lambda}\@lamuarg#1:\@endlamuarg\@ifnextchar\bgroup{.\,\lamu}{.\,}}
\def\@lamuarg#1:#2\@endlamuarg{#1}

\def\fall#1{\forall (#1)\@ifnextchar\bgroup{.\,\fall}{.\,}}

\def\exis#1{\exists (#1)\@ifnextchar\bgroup{.\,\exis}{.\,}}

\def\smsym{\textstyle\sum}
\def\sm#1{\@ifnextchar\bgroup{\sm@parens{#1}}{%
    \@ifnextchar\prd{\sm@parens{#1}\@eatprd}{%
    \@ifnextchar\sm{\sm@parens{#1}\@eatsm}{%
    \@ifnextchar\;{\sm@parens{#1}\@eatsemicolonspace}{%
    \@ifnextchar\\{\sm@parens{#1}\@eatlinebreak}{%
    \@ifnextchar\narrowbreak{\sm@parens{#1}\@eatnarrowbreak}{%
        \sm@noparens{#1}}}}}}}}
\def\sm@parens#1{\@ifnextchar\bgroup%
  {\mathchoice{\@dsm{#1}}{\@tsm{#1}}{\@tsm{#1}}{\@tsm{#1}}\sm@parens}%
  {\@ifnextchar\prd%
    {\mathchoice{\@dsm{#1}}{\@tsm{#1}}{\@tsm{#1}}{\@tsm{#1}}\@eatprd}%
    {\mathchoice{\@dsm{#1}}{\@tsm{#1}}{\@tsm{#1}}{\@tsm{#1}}}}}
\def\@eatprd\prd{\prd@parens}
\def\sm@noparens#1{\mathchoice{\@dsm@noparens{#1}}{\@tsm{#1}}{\@tsm{#1}}{\@tsm{#1}}}
\def\lsm#1{\@ifnextchar\bgroup{\@lsm{#1}\lsm}{\@@lsm{#1}}}
\def\@lsm#1{\mathchoice{{\textstyle\sum}}{\sum}{\sum}{\sum}({\textstyle #1})\;}
\def\@@lsm#1{\mathchoice{{\textstyle\sum}}{\sum}{\sum}{\sum}({\textstyle #1}),\ }
\def\tsm#1{\@tsm{#1}\@ifnextchar\bgroup{\tsm}{}}
\def\@tsm#1{\mathchoice{{\textstyle\sum_{(#1)}}}{\sum_{(#1)}}{\sum_{(#1)}}{\sum_{(#1)}}}
\def\dsm#1{\@dsm{#1}\@ifnextchar\bgroup{\dsm}{}}
\def\@dsm#1{\sum_{(#1)}\,}
\def\@dsm@noparens#1{\sum_{#1}\,}

\def\wtypesym{{\mathsf{W}}}
\def\wtype#1{\@ifnextchar\bgroup%
  {\mathchoice{\@twtype{#1}}{\@twtype{#1}}{\@twtype{#1}}{\@twtype{#1}}\wtype}%
  {\mathchoice{\@twtype{#1}}{\@twtype{#1}}{\@twtype{#1}}{\@twtype{#1}}}}
\def\lwtype#1{\@ifnextchar\bgroup{\@lwtype{#1}\lwtype}{\@@lwtype{#1}}}
\def\@lwtype#1{\mathchoice{{\textstyle\mathsf{W}}}{\mathsf{W}}{\mathsf{W}}{\mathsf{W}}({\textstyle #1})\;}
\def\@@lwtype#1{\mathchoice{{\textstyle\mathsf{W}}}{\mathsf{W}}{\mathsf{W}}{\mathsf{W}}({\textstyle #1}),\ }
\def\twtype#1{\@twtype{#1}\@ifnextchar\bgroup{\twtype}{}}
\def\@twtype#1{\mathchoice{{\textstyle\mathsf{W}_{(#1)}}}{\mathsf{W}_{(#1)}}{\mathsf{W}_{(#1)}}{\mathsf{W}_{(#1)}}}
\def\dwtype#1{\@dwtype{#1}\@ifnextchar\bgroup{\dwtype}{}}
\def\@dwtype#1{\mathsf{W}_{(#1)}\,}

\newcommand{\suppsym}{{\mathsf{sup}}}
\newcommand{\supp}{\ensuremath\suppsym\xspace}

\def\wtypeh#1{\@ifnextchar\bgroup%
  {\mathchoice{\@lwtypeh{#1}}{\@twtypeh{#1}}{\@twtypeh{#1}}{\@twtypeh{#1}}\wtypeh}%
  {\mathchoice{\@@lwtypeh{#1}}{\@twtypeh{#1}}{\@twtypeh{#1}}{\@twtypeh{#1}}}}
\def\lwtypeh#1{\@ifnextchar\bgroup{\@lwtypeh{#1}\lwtypeh}{\@@lwtypeh{#1}}}
\def\@lwtypeh#1{\mathchoice{{\textstyle\mathsf{W}^h}}{\mathsf{W}^h}{\mathsf{W}^h}{\mathsf{W}^h}({\textstyle #1})\;}
\def\@@lwtypeh#1{\mathchoice{{\textstyle\mathsf{W}^h}}{\mathsf{W}^h}{\mathsf{W}^h}{\mathsf{W}^h}({\textstyle #1}),\ }
\def\twtypeh#1{\@twtypeh{#1}\@ifnextchar\bgroup{\twtypeh}{}}
\def\@twtypeh#1{\mathchoice{{\textstyle\mathsf{W}^h_{(#1)}}}{\mathsf{W}^h_{(#1)}}{\mathsf{W}^h_{(#1)}}{\mathsf{W}^h_{(#1)}}}
\def\dwtypeh#1{\@dwtypeh{#1}\@ifnextchar\bgroup{\dwtypeh}{}}
\def\@dwtypeh#1{\mathsf{W}^h_{(#1)}\,}

\makeatother

\let\setof\Set    
\newcommand{\pair}{\ensuremath{\mathsf{pair}}\xspace}
\newcommand{\tup}[2]{(#1,#2)}
\newcommand{\proj}[1]{\ensuremath{\mathsf{pr}_{#1}}\xspace}
\newcommand{\fst}{\ensuremath{\proj1}\xspace}
\newcommand{\snd}{\ensuremath{\proj2}\xspace}
\newcommand{\ac}{\ensuremath{\mathsf{ac}}\xspace} 
\newcommand{\un}{\ensuremath{\mathsf{upun}}\xspace} 

\newcommand{\ind}[1]{\mathsf{ind}_{#1}}
\newcommand{\indid}[1]{\ind{=_{#1}}} 
\newcommand{\indidb}[1]{\ind{=_{#1}}'} 

\newcommand{\uppt}{\ensuremath{\mathsf{uppt}}\xspace}

\newcommand{\pairpath}{\ensuremath{\mathsf{pair}^{\mathord{=}}}\xspace}
\newcommand{\projpath}[1]{\ensuremath{\apfunc{\proj{#1}}}\xspace}

\newcommand{\pairr}[1]{{\mathopen{}(#1)\mathclose{}}}
\newcommand{\Pairr}[1]{{\mathopen{}\left(#1\right)\mathclose{}}}

\newcommand{\im}{\ensuremath{\mathsf{im}}} 

\newcommand{\leftwhisker}{\mathbin{{\ct}_{\mathsf{l}}}}  
\newcommand{\rightwhisker}{\mathbin{{\ct}_{\mathsf{r}}}} 
\newcommand{\hct}{\star}

\newcommand{\modal}{\ensuremath{\ocircle}}
\let\reflect\modal
\newcommand{\modaltype}{\ensuremath{\type_\modal}}
\newcommand{\mreturn}{\ensuremath{\eta}}
\let\project\mreturn
\newcommand{\ext}{\mathsf{ext}}
\renewcommand{\P}{\ensuremath{\type_{P}}\xspace}


\newcommand{\idsym}{{=}}
\newcommand{\id}[3][]{\ensuremath{#2 =_{#1} #3}\xspace}
\newcommand{\idtype}[3][]{\ensuremath{\mathsf{Id}_{#1}(#2,#3)}\xspace}
\newcommand{\idtypevar}[1]{\ensuremath{\mathsf{Id}_{#1}}\xspace}
\newcommand{\defid}{\coloneqq}

\newcommand{\dpath}[4]{#3 =^{#1}_{#2} #4}


\newcommand{\refl}[1]{\ensuremath{\mathsf{refl}_{#1}}\xspace}

\newcommand{\ct}{%
  \mathchoice{\mathbin{\raisebox{0.5ex}{$\displaystyle\centerdot$}}}%
             {\mathbin{\raisebox{0.5ex}{$\centerdot$}}}%
             {\mathbin{\raisebox{0.25ex}{$\scriptstyle\,\centerdot\,$}}}%
             {\mathbin{\raisebox{0.1ex}{$\scriptscriptstyle\,\centerdot\,$}}}
}

\newcommand{\opp}[1]{\mathord{{#1}^{-1}}}
\let\rev\opp

\newcommand{\trans}[2]{\ensuremath{{#1}_{*}\mathopen{}\left({#2}\right)\mathclose{}}\xspace}
\let\Trans\trans
\newcommand{\transf}[1]{\ensuremath{{#1}_{*}}\xspace} 
\newcommand{\transfib}[3]{\ensuremath{\mathsf{transport}^{#1}(#2,#3)\xspace}}
\newcommand{\Transfib}[3]{\ensuremath{\mathsf{transport}^{#1}\Big(#2,\, #3\Big)\xspace}}
\newcommand{\transfibf}[1]{\ensuremath{\mathsf{transport}^{#1}\xspace}}

\newcommand{\transtwo}[2]{\ensuremath{\mathsf{transport}^2\mathopen{}\left({#1},{#2}\right)\mathclose{}}\xspace}

\newcommand{\transconst}[3]{\ensuremath{\mathsf{transportconst}}^{#1}_{#2}(#3)\xspace}
\newcommand{\transconstf}{\ensuremath{\mathsf{transportconst}}\xspace}

\newcommand{\mapfunc}[1]{\ensuremath{\mathsf{ap}_{#1}}\xspace} 
\newcommand{\map}[2]{\ensuremath{{#1}\mathopen{}\left({#2}\right)\mathclose{}}\xspace}
\let\Ap\map
\newcommand{\mapdepfunc}[1]{\ensuremath{\mathsf{apd}_{#1}}\xspace} 
\newcommand{\mapdep}[2]{\ensuremath{\mapdepfunc{#1}\mathopen{}\left(#2\right)\mathclose{}}\xspace}
\let\apfunc\mapfunc
\let\ap\map
\let\apdfunc\mapdepfunc
\let\apd\mapdep

\newcommand{\aptwofunc}[1]{\ensuremath{\mathsf{ap}^2_{#1}}\xspace}
\newcommand{\aptwo}[2]{\ensuremath{\aptwofunc{#1}\mathopen{}\left({#2}\right)\mathclose{}}\xspace}
\newcommand{\apdtwofunc}[1]{\ensuremath{\mathsf{apd}^2_{#1}}\xspace}
\newcommand{\apdtwo}[2]{\ensuremath{\apdtwofunc{#1}\mathopen{}\left(#2\right)\mathclose{}}\xspace}

\newcommand{\idfunc}[1][]{\ensuremath{\mathsf{id}_{#1}}\xspace}

\newcommand{\htpy}{\sim}

\newcommand{\bisim}{\sim}       
\newcommand{\eqr}{\sim}         

\newcommand{\eqv}[2]{\ensuremath{#1 \simeq #2}\xspace}
\newcommand{\eqvspaced}[2]{\ensuremath{#1 \;\simeq\; #2}\xspace}
\newcommand{\eqvsym}{\simeq}    
\newcommand{\texteqv}[2]{\ensuremath{\mathsf{Equiv}(#1,#2)}\xspace}
\newcommand{\isequiv}{\ensuremath{\mathsf{isequiv}}}
\newcommand{\qinv}{\ensuremath{\mathsf{qinv}}}
\newcommand{\ishae}{\ensuremath{\mathsf{ishae}}}
\newcommand{\linv}{\ensuremath{\mathsf{linv}}}
\newcommand{\rinv}{\ensuremath{\mathsf{rinv}}}
\newcommand{\biinv}{\ensuremath{\mathsf{biinv}}}
\newcommand{\lcoh}[3]{\mathsf{lcoh}_{#1}(#2,#3)}
\newcommand{\rcoh}[3]{\mathsf{rcoh}_{#1}(#2,#3)}
\newcommand{\hfib}[2]{{\mathsf{fib}}_{#1}(#2)}

\newcommand{\total}[1]{\ensuremath{\mathsf{total}(#1)}}

\newcommand{\UU}{\ensuremath{\mathcal{U}}\xspace}
\let\bbU\UU
\let\type\UU
\newcommand{\typele}[1]{\ensuremath{{#1}\text-\mathsf{Type}}\xspace}
\newcommand{\typeleU}[1]{\ensuremath{{#1}\text-\mathsf{Type}_\UU}\xspace}
\newcommand{\typelep}[1]{\ensuremath{{(#1)}\text-\mathsf{Type}}\xspace}
\newcommand{\typelepU}[1]{\ensuremath{{(#1)}\text-\mathsf{Type}_\UU}\xspace}
\let\ntype\typele
\let\ntypeU\typeleU
\let\ntypep\typelep
\let\ntypepU\typelepU
\newcommand{\set}{\ensuremath{\mathsf{Set}}\xspace}
\newcommand{\setU}{\ensuremath{\mathsf{Set}_\UU}\xspace}
\newcommand{\prop}{\ensuremath{\mathsf{Prop}}\xspace}
\newcommand{\propU}{\ensuremath{\mathsf{Prop}_\UU}\xspace}
\newcommand{\pointed}[1]{\ensuremath{#1_\bullet}}

\newcommand{\card}{\ensuremath{\mathsf{Card}}\xspace}
\newcommand{\ord}{\ensuremath{\mathsf{Ord}}\xspace}
\newcommand{\ordsl}[2]{{#1}_{/#2}}

\newcommand{\ua}{\ensuremath{\mathsf{ua}}\xspace} 
\newcommand{\idtoeqv}{\ensuremath{\mathsf{idtoeqv}}\xspace}
\newcommand{\univalence}{\ensuremath{\mathsf{univalence}}\xspace} 

\newcommand{\iscontr}{\ensuremath{\mathsf{isContr}}}
\newcommand{\contr}{\ensuremath{\mathsf{contr}}} 
\newcommand{\isset}{\ensuremath{\mathsf{isSet}}}
\newcommand{\isprop}{\ensuremath{\mathsf{isProp}}}

\let\hfiber\hfib

\newcommand{\trunc}[2]{\mathopen{}\left\Vert #2\right\Vert_{#1}\mathclose{}}
\newcommand{\ttrunc}[2]{\bigl\Vert #2\bigr\Vert_{#1}}
\newcommand{\Trunc}[2]{\Bigl\Vert #2\Bigr\Vert_{#1}}
\newcommand{\truncf}[1]{\Vert \blank \Vert_{#1}}
\newcommand{\tproj}[3][]{\mathopen{}\left|#3\right|_{#2}^{#1}\mathclose{}}
\newcommand{\tprojf}[2][]{|\blank|_{#2}^{#1}}
\def\pizero{\trunc0}

\newcommand{\brck}[1]{\trunc{}{#1}}
\newcommand{\bbrck}[1]{\ttrunc{}{#1}}
\newcommand{\Brck}[1]{\Trunc{}{#1}}
\newcommand{\bproj}[1]{\tproj{}{#1}}
\newcommand{\bprojf}{\tprojf{}}

\newcommand{\Parens}[1]{\Bigl(#1\Bigr)}

\let\extendsmb\ext
\newcommand{\extend}[1]{\extendsmb(#1)}

%
\newcommand{\emptyt}{\ensuremath{\mathbf{0}}\xspace}

\newcommand{\unit}{\ensuremath{\mathbf{1}}\xspace}
\newcommand{\ttt}{\ensuremath{\star}\xspace}

\newcommand{\bool}{\ensuremath{\mathbf{2}}\xspace}
\newcommand{\btrue}{{1_{\bool}}}
\newcommand{\bfalse}{{0_{\bool}}}

\newcommand{\inlsym}{{\mathsf{inl}}}
\newcommand{\inrsym}{{\mathsf{inr}}}
\newcommand{\inl}{\ensuremath\inlsym\xspace}
\newcommand{\inr}{\ensuremath\inrsym\xspace}

\newcommand{\seg}{\ensuremath{\mathsf{seg}}\xspace}

\newcommand{\freegroup}[1]{F(#1)}
\newcommand{\freegroupx}[1]{F'(#1)} 

\newcommand{\glue}{\mathsf{glue}}

\newcommand{\colim}{\mathsf{colim}}
\newcommand{\inc}{\mathsf{inc}}
\newcommand{\cmp}{\mathsf{cmp}}

\newcommand{\Sn}{\mathbb{S}}
\newcommand{\base}{\ensuremath{\mathsf{base}}\xspace}
\newcommand{\lloop}{\ensuremath{\mathsf{loop}}\xspace}
\newcommand{\surf}{\ensuremath{\mathsf{surf}}\xspace}

\newcommand{\susp}{\Sigma}
\newcommand{\north}{\mathsf{N}}
\newcommand{\south}{\mathsf{S}}
\newcommand{\merid}{\mathsf{merid}}

\newcommand{\blank}{\mathord{\hspace{1pt}\text{--}\hspace{1pt}}}

\newcommand{\nameless}{\mathord{\hspace{1pt}\underline{\hspace{1ex}}\hspace{1pt}}}

\newcommand{\bbP}{\ensuremath{\mathbb{P}}\xspace}

\newcommand{\uset}{\ensuremath{\mathcal{S}et}\xspace}
\newcommand{\ucat}{\ensuremath{{\mathcal{C}at}}\xspace}
\newcommand{\urel}{\ensuremath{\mathcal{R}el}\xspace}
\newcommand{\uhilb}{\ensuremath{\mathcal{H}ilb}\xspace}
\newcommand{\utype}{\ensuremath{\mathcal{T}\!ype}\xspace}

\newbox\pbbox
\setbox\pbbox=\hbox{\xy \POS(65,0)\ar@{-} (0,0) \ar@{-} (65,65)\endxy}
\def\pb{\save[]+<3.5mm,-3.5mm>*{\copy\pbbox} \restore}

\newcommand{\inv}[1]{{#1}^{-1}}
\newcommand{\idtoiso}{\ensuremath{\mathsf{idtoiso}}\xspace}
\newcommand{\isotoid}{\ensuremath{\mathsf{isotoid}}\xspace}
\newcommand{\op}{^{\mathrm{op}}}
\newcommand{\y}{\ensuremath{\mathbf{y}}\xspace}
\newcommand{\dgr}[1]{{#1}^{\dagger}}
\newcommand{\unitaryiso}{\mathrel{\cong^\dagger}}
\newcommand{\cteqv}[2]{\ensuremath{#1 \simeq #2}\xspace}
\newcommand{\cteqvsym}{\simeq}     

\newcommand{\N}{\ensuremath{\mathbb{N}}\xspace}
\let\nat\N
\newcommand{\natp}{\ensuremath{\nat'}\xspace} 

\newcommand{\zerop}{\ensuremath{0'}\xspace}   
\newcommand{\suc}{\mathsf{succ}}
\newcommand{\sucp}{\ensuremath{\suc'}\xspace} 
\newcommand{\add}{\mathsf{add}}
\newcommand{\ack}{\mathsf{ack}}
\newcommand{\ite}{\mathsf{iter}}
\newcommand{\assoc}{\mathsf{assoc}}
\newcommand{\dbl}{\ensuremath{\mathsf{double}}}
\newcommand{\dblp}{\ensuremath{\dbl'}\xspace} 

\newcommand{\lst}[1]{\mathsf{List}(#1)}
\newcommand{\nil}{\mathsf{nil}}
\newcommand{\cons}{\mathsf{cons}}
\newcommand{\lost}[1]{\mathsf{Lost}(#1)}

\newcommand{\vect}[2]{\ensuremath{\mathsf{Vec}_{#1}(#2)}\xspace}

\newcommand{\Z}{\ensuremath{\mathbb{Z}}\xspace}
\newcommand{\Zsuc}{\mathsf{succ}}
\newcommand{\Zpred}{\mathsf{pred}}

\newcommand{\Q}{\ensuremath{\mathbb{Q}}\xspace}

\newcommand{\funext}{\mathsf{funext}}
\newcommand{\happly}{\mathsf{happly}}

\newcommand{\com}[3]{\mathsf{swap}_{#1,#2}(#3)}

\newcommand{\code}{\ensuremath{\mathsf{code}}\xspace}
\newcommand{\encode}{\ensuremath{\mathsf{encode}}\xspace}
\newcommand{\decode}{\ensuremath{\mathsf{decode}}\xspace}

\newcommand{\function}[4]{\left\{\begin{array}{rcl}#1 &
      \longrightarrow & #2 \\ #3 & \longmapsto & #4 \end{array}\right.}

\newcommand{\cone}[2]{\mathsf{cone}_{#1}(#2)}
\newcommand{\cocone}[2]{\mathsf{cocone}_{#1}(#2)}
\newcommand{\composecocone}[2]{#1\circ#2}
\newcommand{\composecone}[2]{#2\circ#1}
\newcommand{\Ddiag}{\mathscr{D}}

\newcommand{\Map}{\mathsf{Map}}

\newcommand{\interval}{\ensuremath{I}\xspace}
\newcommand{\izero}{\ensuremath{0_{\interval}}\xspace}
\newcommand{\ione}{\ensuremath{1_{\interval}}\xspace}

\newcommand{\epi}{\ensuremath{\twoheadrightarrow}}
\newcommand{\mono}{\ensuremath{\rightarrowtail}}

\newcommand{\bin}{\ensuremath{\mathrel{\widetilde{\in}}}}

\newcommand{\semigroupstrsym}{\ensuremath{\mathsf{SemigroupStr}}}
\newcommand{\semigroupstr}[1]{\ensuremath{\mathsf{SemigroupStr}}(#1)}
\newcommand{\semigroup}[0]{\ensuremath{\mathsf{Semigroup}}}

\newcommand{\emptyctx}{\ensuremath{\cdot}}
\newcommand{\production}{\vcentcolon\vcentcolon=}
\newcommand{\conv}{\downarrow}
\newcommand{\ctx}{\ensuremath{\mathsf{ctx}}}
\newcommand{\wfctx}[1]{#1\ \ctx}
\newcommand{\oftp}[3]{#1 \vdash #2 : #3}
\newcommand{\jdeqtp}[4]{#1 \vdash #2 \jdeq #3 : #4}
\newcommand{\judg}[2]{#1 \vdash #2}
\newcommand{\tmtp}[2]{#1 \mathord{:} #2}

\newcommand{\form}{\textsc{form}}
\newcommand{\intro}{\textsc{intro}}
\newcommand{\elim}{\textsc{elim}}
\newcommand{\uniq}{\textsc{uniq}}
\newcommand{\Weak}{\mathsf{Wkg}}
\newcommand{\Vble}{\mathsf{Vble}}
\newcommand{\Exch}{\mathsf{Exch}}
\newcommand{\Subst}{\mathsf{Subst}}

\newcommand{\cc}{\mathsf{c}}
\newcommand{\pp}{\mathsf{p}}
\newcommand{\cct}{\widetilde{\mathsf{c}}}
\newcommand{\ppt}{\widetilde{\mathsf{p}}}
\newcommand{\Wtil}{\ensuremath{\widetilde{W}}\xspace}

\newcommand{\istype}[1]{\mathsf{is}\mbox{-}{#1}\mbox{-}\mathsf{type}}
\newcommand{\nplusone}{\ensuremath{(n+1)}}
\newcommand{\nminusone}{\ensuremath{(n-1)}}
\newcommand{\fact}{\mathsf{fact}}

\newcommand{\kbar}{\overline{k}} 

\newcommand{\natw}{\ensuremath{\mathbf{N^w}}\xspace}
\newcommand{\zerow}{\ensuremath{0^\mathbf{w}}\xspace}
\newcommand{\sucw}{\ensuremath{\mathsf{succ}^{\mathbf{w}}}\xspace}
\newcommand{\nalg}{\nat\mathsf{Alg}}
\newcommand{\nhom}{\nat\mathsf{Hom}}
\newcommand{\ishinitw}{\mathsf{isHinit}_{\mathsf{W}}}
\newcommand{\ishinitn}{\mathsf{isHinit}_\nat}
\newcommand{\w}{\mathsf{W}}
\newcommand{\walg}{\w\mathsf{Alg}}
\newcommand{\whom}{\w\mathsf{Hom}}

\newcommand{\RC}{\ensuremath{\mathbb{R}_\mathsf{c}}\xspace} 
\newcommand{\RD}{\ensuremath{\mathbb{R}_\mathsf{d}}\xspace} 
\newcommand{\R}{\ensuremath{\mathbb{R}}\xspace}           
\newcommand{\barRD}{\ensuremath{\bar{\mathbb{R}}_\mathsf{d}}\xspace} 

\newcommand{\close}[1]{\sim_{#1}} 
\newcommand{\closesym}{\mathord\sim}
\newcommand{\rclim}{\mathsf{lim}} 
\newcommand{\rcrat}{\mathsf{rat}} 
\newcommand{\rceq}{\mathsf{eq}_{\RC}} 
\newcommand{\CAP}{\mathcal{C}}    
\newcommand{\Qp}{\Q_{+}}
\newcommand{\apart}{\mathrel{\#}}  
\newcommand{\dcut}{\mathsf{isCut}}  
\newcommand{\cover}{\triangleleft} 
\newcommand{\intfam}[3]{(#2, \lam{#1} #3)} 

\newcommand{\bsim}{\frown}
\newcommand{\bbsim}{\smile}

\newcommand{\hapx}{\diamondsuit\approx}
\newcommand{\hapname}{\diamondsuit}
\newcommand{\hapxb}{\heartsuit\approx}
\newcommand{\hapbname}{\heartsuit}
\newcommand{\tap}[1]{\bullet\approx_{#1}\triangle}
\newcommand{\tapname}{\triangle}
\newcommand{\tapb}[1]{\bullet\approx_{#1}\square}
\newcommand{\tapbname}{\square}

\newcommand{\NO}{\ensuremath{\mathsf{No}}\xspace}
\newcommand{\surr}[2]{\{\,#1\,\big|\,#2\,\}}
\newcommand{\LL}{\mathcal{L}}
\newcommand{\RR}{\mathcal{R}}
\newcommand{\noeq}{\mathsf{eq}_{\NO}} 

\newcommand{\ble}{\trianglelefteqslant}
\newcommand{\blt}{\vartriangleleft}
\newcommand{\bble}{\sqsubseteq}
\newcommand{\bblt}{\sqsubset}

\newcommand{\hle}{\diamondsuit\preceq}
\newcommand{\hlt}{\diamondsuit\prec}
\newcommand{\hlname}{\diamondsuit}
\newcommand{\hleb}{\heartsuit\preceq}
\newcommand{\hltb}{\heartsuit\prec}
\newcommand{\hlbname}{\heartsuit}
\newcommand{\tle}{\triangle\preceq}
\newcommand{\tlt}{\triangle\prec}
\newcommand{\tlname}{\triangle}
\newcommand{\tleb}{\square\preceq}
\newcommand{\tltb}{\square\prec}
\newcommand{\tlbname}{\square}

\newcommand{\vset}{\mathsf{set}}  
\def\cd{\tproj0}
\newcommand{\inj}{\ensuremath{\mathsf{inj}}} 
\newcommand{\acc}{\ensuremath{\mathsf{acc}}} 

\newcommand{\atMostOne}{\mathsf{atMostOne}}

\newcommand{\power}[1]{\mathcal{P}(#1)} 
\newcommand{\powerp}[1]{\mathcal{P}_+(#1)} 

\setcounter{secnumdepth}{5}

\crefname{figure}{Figure}{Figures}


\numberwithin{equation}{section}


\renewcommand{\theenumi}{(\roman{enumi})}
\renewcommand{\labelenumi}{\theenumi}




\def\noteson{%
\gdef\note##1{\mbox{}\marginpar{\color{blue}\textasteriskcentered\ ##1}}}
\gdef\notesoff{\gdef\note##1{\null}}
\noteson

\newcommand{\Coq}{\textsc{Coq}\xspace}
\newcommand{\Agda}{\textsc{Agda}\xspace}
\newcommand{\NuPRL}{\textsc{NuPRL}\xspace}




\newcommand{\footstyle}[1]{{\hyperpage{#1}}n} 
\newcommand{\defstyle}[1]{\textbf{\hyperpage{#1}}}  

\newcommand{\indexdef}[1]{\index{#1|defstyle}}   
\newcommand{\indexfoot}[1]{\index{#1|footstyle}} 
\newcommand{\indexsee}[2]{\index{#1|see{#2}}}    


\newcommand{\ZF}{Zermelo--Fraenkel}
\newcommand{\CZF}{Constructive \ZF{} Set Theory}

\newcommand{\LEM}[1]{\ensuremath{\mathsf{LEM}_{#1}}\xspace}
\newcommand{\choice}[1]{\ensuremath{\mathsf{AC}_{#1}}\xspace}


\newcommand{\mentalpause}{\medskip} 

\newcounter{symindex}
\newcommand{\symlabel}[1]{\refstepcounter{symindex}\label{#1}}


\title{Heterogeneous substitution systems revisited}

\author{Benedikt Ahrens}
\author{Ralph Matthes}

\thanks{The work of Benedikt Ahrens was partially supported by the CIMI (Centre International de Mathématiques et d'Informatique)
Excellence program ANR-11-LABX-0040-CIMI within the program ANR-11-IDEX-0002-02 during a postdoctoral fellowship.
\\
This material is based upon work supported by the National Science Foundation under agreement Nos. DMS-1128155 and CMU 1150129-338510.
Any opinions, findings and conclusions or recommendations expressed in this material are those of the author(s) and do not necessarily reflect the views of the National Science Foundation.}

\begin{abstract}
 Matthes and Uustalu (TCS 327(1--2):155--174, 2004) presented a
 categorical description of substitution systems capable of capturing
 syntax involving binding which is independent of whether the syntax
 is made up from least or greatest fixed points.
 We extend this work
 in two directions: we continue the analysis by creating more
 categorical structure, in particular by organizing substitution
 systems into a category and studying its properties, and we develop
 the proofs of the results of the cited paper and our new ones in
 \UniMath, a recent library of univalent mathematics formalized in the \Coq theorem
 prover.
\end{abstract}

\maketitle

\tableofcontents

\section{Introduction}
Given a first-order signature over some supply of variables,
substitution is nearly a homomorphism: the substitution function
commutes with all term-forming operations (however, at leaf positions,
variables may get replaced by terms). But substitution also gives rise
to a monad structure. For this, it is useful to see the variable
supply of the terms as a parameter: writing $TA$ for the set of terms
over variable supply $A$ (those variables that may occur free in the
terms), parallel substitution associates with each substitution rule
$f$, which is a function from $A$ to $TB$, a substitution function
$[f]:TA\to TB$, and for a given term $t:TA$, the term $t[f]:TB$
(notice the post-fix notation for function $[f]$) is the result of the
parallel substitution that replaces each occurrence of a variable
$x:A$ in $t$ by $fx:TB$. In fact, the function $T$, the function that
injects variables into terms, and the operation of parallel
substitution together form a monad in the format of a Kleisli triple
over the category of sets and functions. Notice that the types serve
as a means of tracking the (names of) variables that may occur free in
a term, the object syntax itself is untyped. The parameter $A$ plays a
more prominent role as soon as variable binding is allowed in the
object syntax: for pure $\lambda$-calculus, bound and free variable
occurrences have to be distinguished, and even the constructors of the
object language relate terms with different variable supply, in
particular $\lambda$-abstraction assumes an argument term where the
newly bound variable is added to the variable supply (this will be
seen with more details in
Section~\ref{sec:explicit_flattening}.). Although parallel
substitution $t[f]$ has to be defined with extra care to avoid capture
of free variables of some $fx$ by binders in $t$, it is still (modulo
$\alpha$-equivalence) nearly a homomorphism, and it still yields a
monad \cite{BellegardeHook}. However, the monad laws by themselves do
not express the (nearly) ``homomorphic nature'' of substitution.

In previous work, Matthes and Uustalu
\cite{DBLP:journals/tcs/MatthesU04} define a notion of
\enquote{heterogeneous substitution system}, the purpose of which is
to axiomatize substitution and its desired properties.  Such a
substitution system is given by an algebra of a signature functor,
equipped with an operation---which is to be thought of as
substitution---that is compatible with the algebra structure map in a
suitable sense.  The term \enquote{heterogeneous} refers to the fact
that the underlying notion of signature encompasses variable binding
constructions and also explicit substitution a.\,k.\,a.\ flattening.
More precisely, the signature is based on a rank-2 functor $H$ (an
endofunctor on a category of endofunctors) for the respective
domain-specific signature, to which a monadic unit is explicitly
added.  The latter corresponds to the inclusion of variables into the
elements that are considered as terms (in a quite general sense) over
their variable supply.
The name \enquote{rank-2 functor} stems from
the rank of the type operator that transforms type transformations
into type transformations---hence has kind
$(\Set\to\Set)\to(\Set\to\Set)$---which may be seen as backbone of $H$
in case the base category is $\Set$.
In this rank-2 setting, the carrier of the algebra is an
endofunctor, and since a monadic unit is already present, a natural
question is if one obtains a monad.  In that paper, it is then shown
that for any heterogeneous substitution system this is indeed the
case; multiplication of the monad is derived from the
\enquote{substitution} operation which is parameterized by a morphism
$f$ of pointed endofunctors and consists in asking for a unique
solution that makes a certain diagram commute.  Monad multiplication
and one of the monad laws is obtained from the existence of a solution
in the case that $f$ is the identity, while the other monad laws are
derived from uniqueness for two other choices of $f$.

Furthermore, it is shown there that \enquote{substitution is for free}
for both initial algebras as well as---maybe more surprisingly---for
(the inverse of) final coalgebras: if the initial algebra,
resp.\ terminal coalgebra, of a given signature functor exists, then
it, resp.\ its inverse, can be augmented to a substitution system (for
the former case, and in order to easily use generalized iteration
\cite{DBLP:journals/fac/BirdP99}, it is assumed that the functor
$-\cdot Z$ has a right adjoint for every endofunctor $Z$).  Indeed, it
was one of the design goals of the axiomatic framework of
heterogeneous substitution systems to be applicable to
\emph{non-wellfounded} syntax as well as to wellfounded syntax,
whereas related work (e.g., \cite{fpt,ahrens_ext_init}) frequently
only applies to wellfounded syntax.

Examples of substitution systems are thus given by the lambda calculus, with and without explicit flattening, but also
by languages involving typing and \emph{infinite} terms.

The goal of the present work is twofold:

Firstly, we extend the work by
Matthes and Uustalu \cite{DBLP:journals/tcs/MatthesU04}; in
particular, we introduce a natural notion of \emph{morphisms} of
heterogeneous substitution systems, thus arranging them into a
category.  We then show that the construction of a monad from a
heterogeneous substitution system from
\cite{DBLP:journals/tcs/MatthesU04} extends functorially to morphisms.
Moreover, we prove that the substitution system obtained in
\cite{DBLP:journals/tcs/MatthesU04} by equipping the initial algebra
with a substitution operation, is initial in the corresponding
category of substitution systems. This makes use of a general fusion
law for generalized iteration
\cite{DBLP:journals/fac/BirdP99}.
Moreover, we prove that the property of being initial in the category of algebras lifts to initiality of the
associated substitution system in the corresponding category.
As an example of the usefulness of our results, we express
the resolution of explicit flattening of the lambda calculus as a(n
initial) morphism of substitution systems.

A second part of our work is the formalization of
some of our results in univalent
foundations, more specifically, building upon the \UniMath library
\cite{unimath}. This basis of our formalization is suitable in that it
provides extensionality (functional and
propositional) in a natural way and hereby avoids the use of setoids that would otherwise
be inevitable; indeed, since our results are not about categories \emph{in
  abstracto} but use general categorical concepts in more concrete
instances such as the endofunctor category over a given category or
its extension by a \enquote{point}, we need extensionality axioms for the instantiation.
We
profit from the existing category theory library \cite{rezk_completion} in \UniMath.

\subsection*{Related work}

Related work is extensively discussed in Matthes and Uustalu's article \cite{DBLP:journals/tcs/MatthesU04}.

In the meantime, monads and modules over monads, have been used by 
Hirscho\-witz and Maggesi \cite{DBLP:conf/wollic/HirschowitzM07,DBLP:journals/iandc/HirschowitzM10}
to define models of syntax, and to give a categorical characterization thereof.

The notion of signature introduced in \cite{DBLP:journals/tcs/MatthesU04} and formalized in the 
present work is similar to that employed in Hirschowitz and Maggesi's 
most recent work \cite{DBLP:journals/corr/abs-1202-3499}.
One difference is that we do not, in the present work, insist on our signature functor to be $\omega$-cocontinuous,
since we do not worry about the existence of initial algebras, but assume them to exist.
In our follow-up work with Mörtberg \cite{ahrens_mortberg} on the construction
of initial algebras in sets, however, this condition will be of the essence.

Monads and modules over monads can also be used as the basis, the \enquote{raw syntax}, 
from which dependently typed theories are carved out, as exhibited by Voevodsky \cite{csystem_monad}.
Our formalization provides one of the many steps involved, providing a monad structure on an initial algebra
of a rank-2 endofunctor.

\subsection*{Synopsis}

In Section \ref{sec:univalent_mathematics} we first give a brief overview of
the univalent foundations we work in.
Afterwards, we review the definition of categories in those foundations,
and finally, we show how the foundations are realized in the proof assistant \Coq.

In Section \ref{sec:preliminaries} we define a few basic concepts and introduce notation.

In Section \ref{sec:mendler} we present \enquote{Generalized Iteration in Mendler-style}, and
a fusion law satisfied by this form of iteration.
The presented results will be used in Section \ref{sec:initiality}.

In Section \ref{sec:hss} we review the notion of heterogeneous substitution system.
Afterwards, we define a category of substitution systems and prove a few
properties about that category.

In Section \ref{sec:monads} we state one of the main results of \cite{DBLP:journals/tcs/MatthesU04},
the construction of a monad from a substitution system.
We then prove that the map thus constructed extends to morphisms and yields a faithful functor.

In Section \ref{sec:initiality} we state another of the important results of
\cite{DBLP:journals/tcs/MatthesU04}: the construction of a substitution system from
an initial algebra via Generalized Iteration in Mendler-style as presented in Section \ref{sec:mendler}.
We show that the obtained substitution system is again initial, using
the fusion law stated in \ref{sec:mendler}.

In Section \ref{sec:explicit_flattening}, we construct a particular morphism of
substitution systems, the underlying map of which \enquote{computes away}
explicit substitution of lambda calculus.

Most of the results presented in this article,
both by Matthes and Uustalu \cite{DBLP:journals/tcs/MatthesU04} and our new results, have been formalized, based on the
\UniMath library \cite{unimath}.
More precisely, all results except for
Theorem \ref{thm:hss_sat} and Lemmas \ref{lem:sub_sat} and \ref{lem:hss_replete} are proved in our formalization;
Section~\ref{sec:formalization} provides some technical details about our library.

\section{Univalent mathematics}\label{sec:univalent_mathematics}
The original article \cite{DBLP:journals/tcs/MatthesU04} is written without referring to a specific foundation of mathematics.
Indeed, the authors use purely categorical methods to derive their results.

Our analysis and continuation of that article takes place in a \emph{type-theoretic} foundation, more specifically,
in a type theory augmented by Voevodsky's \emph{Univalence Axiom}.
The resulting theory, to which we refer by the name \enquote{HoTT} in this article, is extensively described elsewhere \cite{hottbook};
we do not attempt to give a comprehensive introduction to HoTT or to the Univalence Axiom in this article.
Instead, here we focus on some of the salient features of HoTT and indicate why they are important to us.

\subsection{About univalent foundations}\label{sec:uf}

  By \enquote{univalent foundations} we refer to an intensional Martin-Löf type theory (IMLTT)
  augmented by Voevodsky's univalence axiom.
  In the following, we give a brief overview of the type constructors available in univalent foundations,
  and a technical statement of the univalence axiom.

   Technically, the univalent foundation we work in is a dependent type theory.
   For a dependent type $B$ over $A$, written $x : A \entails B(x)$, there is the \fat{dependent sum} $\sm{x:A}B(x)$,
   elements of which are dependent pairs $(a,p)$ where $a : A$ and $p : B(a)$.
   The type $\prd{x:A}B(x)$ is the type of dependent functions from $A$ to $B$, that is, a function $f : \prd{x:A}B(x)$
   maps $a:A$ into the type $B(a)$.

   Special, non-dependent, cases of the aforementioned constructors are the cartesian product $A \times B$ and
   the function type $A \to B$.

   For any type $A$ and $a,b : A$ elements of $A$, there is the Martin-Löf identity type $a =_{A} b$ of \enquote{(propositional) equalities} between $a$ and $b$.
   We often omit the subscript $A$ and hence simply write $a = b$.

   One of the most salient features of univalent foundations is the univalence axiom.
   Intuitively, it says that any construction expressible in intensional type theory is
   invariant under equivalence of types.
   What is equivalence of types? The reader can think of it as isomorphism of types:
   two types $A$ and $B$ are isomorphic if there are maps $f : A \to B$ and $g : B \to A$ such
   that both composites $f \circ g$ and $g \circ f$ are pointwise equal (with respect to propositional equality)
   to the identity function. While the definition of equivalence is more refined than that of an isomorphism of types,
   it is the case that any isomorphism gives rise to an equivalence, that is, two types are isomorphic
   if and only if they are equivalent.
   The univalence axiom is stated for a particular given universe. Define, for a fixed universe $\type$,
   the canonical map
   \[
     \idtoeqv : \prd{A,B:\type} A = B \to A \simeq B
   \]
   from identities to equivalences between $A$ and $B$; it is defined by identity elimination,
   mapping the reflexivity term $\refl{A} : A = A$ to the identity equivalence on $A$.
   The universe $\type$ is called \fat{univalent} if for any $A$ and $B$ in $\type$, the map
   $\idtoeqv_{A,B}$ is an equivalence.

  The univalence axiom has a number of desirable consequences---provable \fat{inside} the theory---which can
  be subsumed by the term \enquote{equivalence principle}:
  The equivalence principle says, intuitively, that reasoning about mathematical objects should be
  invariant under an appropriate notion of \enquote{equivalence}
  for those objects.
  In the  foundation we work in, the equivalence principle can be proved
  for function types (function extensionality),
  for mathematical structures such as groups and rings \cite{Coquand20131105},
  and for categories \cite{rezk_completion}.

  A second salient feature of univalent foundations is its \fat{internal} notion of \fat{propositions} and \fat{sets}.
  A type $A$ is called a proposition if it satisfies the (propositional) \enquote{proof irrelevance} principle, that is,
  if one can construct a term of type
  \[
     \isprop(A) := \prd{x,y:A} x = y \enspace .
  \]
  Furthermore, a type $A$ is called a set if all of its identity types are propositions, that is, if one can
  construct a term of type
  \[
     \isset(A) := \prd{x,y:A} \isprop(x = y) \enspace .
  \]
  These two definitions are actually special cases of a more general definition of \fat{homotopy levels} of types.
  However, the general definition will not be of use in this article, and can be consulted in \cite{hottbook}.
  We call \fat{proposition} any type that is a proposition in this sense, that is, any element of $\prop:=\sm{X:\type} \isprop(X)$,
  and similarly for \fat{sets}.

\subsection{Category theory in univalent foundations}\label{sec:ct_in_uf}

  Some category theory in univalent foundations has been developed in \cite{rezk_completion}.
  A category $\C$ is given by
  \begin{itemize}
   \item a type $\C_0$ of objects;
   \item for any $a,b : \C_0$, a type $\C(a,b)$ of morphisms from $a$ to $b$;
   \item for any $a : \C_0$, an identity morphism $\identity(a) : \C(a,a)$;
   \item for any $a,b,c : \C_0$, a composition function $\C(a,b) \to \C(b,c) \to \C(a,c)$,
         written $f \mapsto g \mapsto g \circ f$;
   \item for any $a,b:\C_0$ and $f:\C(a,b)$, we have $f \circ \identity(a) = f$ and $\identity(b) \circ f = f$;
  \item for any $a,b,c,d:A$ and $f:\C(a,b)$, $g:\C(b,c)$, $h:\C(c,d)$, we have ${h\circ (g\circ f)} = {(h\circ g)\circ f}$.
  \end{itemize}

  There is an important difference between categories as usually formalized in
  intensional type theory and categories as considered in \cite{rezk_completion}:
  in intensional type theory, categories are usually defined to come with a custom equivalence relation on the types of morphisms, which is to be read as equality relation on morphisms, 
  specified for each category individually (see, e.g., \cite{ConCat}, \cite[Chapter 6]{ahrens_phd}).
  This notion of category is sometimes referred to by \enquote{E-categories} \cite{ecats_palmgren}.

  In the formalization of \cite{rezk_completion}, which takes place in univalent foundations, however,
  the authors consider morphisms of a category modulo equality as given by the identity type.
  That this is feasible is due to the extensional features that the univalence axiom adds to type theory,
  in particular, function extensionality.

  The notion of category is actually more refined in \cite{rezk_completion};
  two conditions must be satisfied by a category:
  \begin{enumerate}
    \item \label{item:homsets} Its hom-types $\C(a,b)$ need to be \fat{sets}.
     This is necessary for the axioms---which talk about equality of arrows---to be
  propositions.
    \item \label{item:univalent} Secondly, in a category, the type of (propositional) equalities (as given by
  the Martin-Löf identity type) between any two objects must be equivalent to the type of isomorphisms between those objects.
  More precisely, to any category one defines a family of maps
  \[
    \idtoiso: \prd{a,b:\C_0} (a = b) \to \iso(a,b) \enspace.
  \]
   This family of maps is defined by identity elimination, mapping $\refl{a} : a = a$ to the identity isomorphism on $a$.
   A category $\C$ is called \fat{univalent}, if for any $a,b : \C_0$, the map $\idtoiso_{a,b}$ is an equivalence.

   \end{enumerate}

  The univalence condition for categories \ref{item:univalent} states, intuitively, that isomorphic objects in such a category cannot be distinguished.
  The equivalence principle for
  univalent categories, proved in \cite{rezk_completion},
  then says that any two equivalent such categories cannot be distinguished either, that is,
  the postulated invariance on objects (univalence) lifts to the categories themselves.
  One of the results proved below shows that our main category of interest is univalent 
  if one starts 
  with a univalent category (Theorem~\ref{thm:hss_sat}).

   An important remark about naming: in \cite{rezk_completion},
   the term \enquote{precategory} is employed for categories that satisfy condition \ref{item:homsets},
   and the term \enquote{category} is reserved for categories that, additionally, satisfy the univalence condition \ref{item:univalent}.
   That is, the authors of \cite{rezk_completion} use the terms
   \enquote{precategory} and \enquote{category} for what we call \enquote{category} and \enquote{univalent category}
   in the present article, respectively.
   The rationale behind this naming convention in \cite{rezk_completion} is that the notion of categories satisfying condition \ref{item:univalent}
   should be considered to be the right notion of category,
   for those categories satisfy the equivalence principle. Furthermore, many important examples of categories do satisfy this condition,
   and the condition is closed under a lot of constructions of new categories from old categories:
   \begin{itemize}
    \item the category of sets and functions between them is univalent;
    \item categories of algebraic structures (groups, rings,\ldots) are univalent;
    \item the functor category $[\C,\D]$ is univalent if $\D$ is;
    \item a full subcategory of a univalent category is again univalent.
   \end{itemize}
   More constructions of categories that preserve univalence are given below.

   For the purposes of the present article, the univalence condition on categories is not essential.
   Indeed, no other result depends on Theorem~\ref{thm:hss_sat}
   We thus choose to de-emphasize the importance of the univalence condition for categories
   by deviating from the naming of \cite{rezk_completion},
   and instead to make it explicit when considering categories that satisfy univalence.

\subsection{About \UniMath}\label{sec:unimath}

  The goal of the \UniMath library is to provide a library of computer-checked mathematics
  formalized in (a computer implementation of) the univalent foundations.
  At this time, there is no computer theorem prover that implements exactly the univalent foundations
  as described in Section \ref{sec:uf}.
  As an approximation for such a tool, we use the \Coq proof assistant \cite{Coq:manual} as a base of \UniMath. 
  However, in order to simulate working in the theory described in Section \ref{sec:uf}, we
  do not use the full language \Coq provides, but restrict ourselves to the language constructors
  described above.
  In particular, there is no use of inductive types besides that of the natural numbers,
  and of the identity type and the type of dependent pairs, both of which are not primitives in \Coq,
  but instead implemented via the general \texttt{Inductive} vernacular.
  Furthermore, record types are not used in \UniMath; bundling of structures is instead implemented via
  (iterated) Sigma types.

  The proof assistant \Coq has recently gained a form of
  universe polymorphism \cite{DBLP:conf/itp/SozeauT14}.
  Unfortunately, this universe management is not powerful enough for our purposes. In particular, it does not
  implement a form of \emph{resizing} rule that is needed for some impredicative encodings of
  constructions---propositional truncation in particular, as described by Voevodsky \cite[Section 4]{voevodsky_experimental}.
  It was thus Voevodsky's choice to use a modified version of \Coq where the checking of
  universe levels was deactivated, and the system hence inconsistent.
  In the meantime, \Coq has been improved to allow the disabling of universe checking via a flag \texttt{-type-in-type}
  passed to the program, instead of modifying its source code.
  The \UniMath library hence is based on an unmodified version of \Coq,
  but is still working in an inconsistent system for now,
  while waiting for a new, more suitable universe management to be implemented.

  Another difference to standard \Coq is our use of the \texttt{-indices-matter} flag.
  This flag ensures that the identity type associated to a type $A$, lives in the same universe
  as the type $A$ itself. By default, without that flag, \Coq would put the identity type
  into the universe \texttt{Prop}
  (not to be confounded with the \fat{homotopy level} of propositions explained in Section \ref{sec:uf}).

  The experimental \enquote{Higher Inductive Types} (HITs),
  described e.g. in the HoTT book \cite{hottbook}, are not used in \UniMath.

  The univalence axiom is implemented in \UniMath via the \texttt{Axiom} vernacular of \Coq.
  This leads to potentially non-normalizing terms, when using the axiom or any of
  its consequences---such as function extensionality.
  We do not experience any problems related to non-normalization,
  since we only use the univalence axiom (indirectly by using function extensionality)
  for proving propositions, not for specifying operations.

\section{Preliminaries}\label{sec:preliminaries}

Categories, functors and natural transformations are defined in \cite{rezk_completion}.
Some more concepts and notation are defined in the following:

For functors $F : \C \to \D$ and $G : \D \to \E$, we write $G \cdot F : \C \to \E$ for their composition.
We use the same notation for composition of a functor with a natural transformation (sometimes called \enquote{whiskering}),
as in $\tau \cdot F$ and $G \cdot \tau$.

\begin{definition}[pointed functors]\label{def:pointed_functor}
  Let \C be a category. We denote by $\Ptd(\C)$ the category of
  pointed endofunctors on \C, an object of which is a pair $(X,\eta)$
  of an endofunctor $X$ on \C and a natural transformation $\eta :
  \Id{\C} \to X$, called a \enquote{point} of $X$, where $\Id{\C}$ is the identity functor on $\C$.
  Morphisms of pointed
  functors are natural transformations between the underlying
  endofunctors that are compatible with the chosen points. Call $U$ the
  forgetful functor from $\Ptd(\C)$ to the underlying endofunctor
  category $[\C,\C]$ (in particular, for a morphism $f$, $Uf$ is $f$, but its compatibility with the points is not taken into account in the type information---justifying to confuse $Uf$ and $f$ in the rest of the paper).
\end{definition}

\begin{definition}[monoidal structure on functor categories]\label{def:monoidal_endofunctors}
The monoidal structure on the endofunctor
  category $[\C,\C]$ given by composition extends to $\Ptd(\C)$.
We denote by $\alpha_{X,Y,Z} : X \cdot (Y \cdot Z) \simeq (X \cdot Y) \cdot Z$, $\rho_X : \Id{\C} \cdot X \simeq X$ and $\lambda_X : X \cdot \Id{\C} \simeq X$ the monoidal isomorphisms.
\end{definition}

\begin{remark}
 In \cite{DBLP:journals/tcs/MatthesU04}, the authors implicitly assume the monoidal structures on $[\C,\C]$ and $\Ptd(\C)$ to be strict.
 In univalent foundations, \enquote{strict} should mean \enquote{the same modulo definitional equality}; 
 the monoidal structures are not strict for this notion of strictness.
 Instead, we need to explicitly insert 
 the isomorphisms (which correspond to propositional equalities in \emph{univalent} categories,
 but that shall not be of importance in the following).
 Note, however, that those isomorphisms are given by families of identity morphisms,
 and thus do not carry any information at all; they are merely needed to 
 formally adjust the type of source and target functors of the natural transformations involved
 in order to allow composing two natural transformations which would not be composable otherwise.
 Indeed, composability of two natural transformations $\alpha : F \to G$ and $\beta: G' \to H$
 depends on $G$ being \emph{definitionally equal} to $G'$.
\end{remark}

\begin{definition}[algebras of a functor]\label{def:algebra}
  For an endofunctor $F : \C \to \C$, the category $\Alg(F)$ of \fat{algebras} has, as objects,
  pairs $(X,\alpha)$ of an object $X : \C_0$ and a morphism $\alpha : \C(FX,X)$.
  For a given algebra $(X,\alpha)$, we call $X$ the \fat{(algebra) carrier} of the algebra.
  A morphism $f : \Alg(F)\bigl((X,\alpha),(X',\alpha')\bigr)$ is given by a morphism
  $f : \C(X,X')$ such that $f \circ \alpha = \alpha' \circ Ff$.
\end{definition}

\begin{convention}
  We are using the arrow symbol \enquote{$\to$} for three different things:
  \begin{enumerate}
   \item morphisms $f : c \to d$ in a category, as shorthand for $f : \C(c,d)$
         (hence in particular for natural transformations as morphisms in functor categories);
   \item functors $F : \C \to \D$ between categories; and
   \item type-theoretic functions $f : A \to B$.
  \end{enumerate}
  Information on what the arrow denotes in each occurrence will be deducible from the context.
\end{convention}

\begin{definition}[monads]\label{def:monads}
  For a category $\C$, the category $\Mon(\C)$ of \fat{monads} has, as objects, triples $(T,\eta,\mu)$ of an endofunctor $T$ of $\C$, and natural transformations $\eta:\Id{\C}\to T$ and $\mu:T\cdot T\to T$ (using our convention on natural transformations), subject to the usual monad laws. A morphism $f:\Mon(\C)\bigl((T,\eta,\mu),(T',\eta',\mu')\bigr)$ is given by a natural transformation $f:T\to T'$, subject to the usual compatibility conditions.
\end{definition}
Notice that we follow \cite{DBLP:journals/tcs/MatthesU04} in taking monad multiplication $\mu$ as third component of a monad and not the binding operation that is more widespread in computer science literature.

\begin{convention}
  Given $d : \D$ and a category $\C$, we call $\constfunctor{d} : \C \to \D$ the functor that
  is constantly $d$ and $\identity_d$ on objects and morphisms, respectively.
  This notation hides the category $\C$, which will usually be deducible from the context.
  In this article, $\C$ will always be $\D$.
\end{convention}

\section{Generalized Iteration in Mendler-style and fusion law}\label{sec:mendler}

In this section we discuss \enquote{generalized iteration in Mendler-style} and a fusion law
that one can prove for this iteration scheme.
Both the iteration scheme and the fusion law are used in Section \ref{sec:initiality}.

\begin{lemma}[Generalized iteration in Mendler-style (Theorem 2 of \cite{DBLP:journals/fac/BirdP99} by Bird and Paterson)]\label{lem:gen_mendler_it}
 Let $\C$ be a category, and let $F : \C\to\C$ be an endofunctor on $\C$.
 Suppose $(\mu F,\In{})$ is the initial algebra of $F$.
 Let $\D$ be another category,
 and let $\C : L\dashv R : \D$ be an adjunction.
 Let $X : \D_0$ be an object of $\D$,
 and let
 \[ \Psi : \D(L -, X) \to \D(L(F -) , X) \]
 be a natural transformation.
 Then there is exactly one morphism $h : L(\mu F) \to X$ such that
 the following diagram commutes:
 \[
  \begin{xy}
    \xymatrix {
                    L(F(\mu F)) \ar[r]^-{L\In{}}  \ar[rd]_{\Psi_{\mu F}(h)}& L(\mu F) \ar[d]^{h} \\
                                &  X
      }
  \end{xy}
 \]
\end{lemma}

  We call $\giter{F}{L}{}{\Psi}:=h$ the unique morphism thus specified.

  \medskip 

Note that, strictly speaking, the functors occurring in the type of $\Psi$ have to be the opposites of $L$ and $F$.

The link with the work by Mendler
\cite{DBLP:journals/apal/Mendler91} is not made in the original
proof \cite{DBLP:journals/fac/BirdP99} of the lemma. The
presentation in \cite{DBLP:journals/fac/BirdP99} is very much
oriented towards functional programming. In their notation, the
natural transformation $\Psi$ would be typed as
$$\Psi::\forall A.\,(LA\to X)\to (L(FA)\to X)\enspace.$$
 
The existence of the right adjoint $R$ for $L$ is rather a matter of
technical convenience: it can be replaced by asking for the
preservation of colimits of chains by $F$ and $L$ and the preservation
of initiality by $L$ \cite[Theorem 1]{DBLP:journals/fac/BirdP99}, but
we do not pursue that alternative in our formalization. 

In \cite{DBLP:journals/tcs/MatthesU04}, only a specialized form of
generalized iteration in Mendler-style is used that is called ``generalized
iteration'' (again with no hint to Mendler's work---see our remarks in
Section~\ref{sec:initiality} on the connection). The specialization
consists in taking only natural transformations $\Psi$ of a specific
form (so that $\Psi$ disappears from the formulation, as explained in
\cite{DBLP:journals/tcs/MatthesU04}). In fact, we do not need the
fuller generality of generalized iteration in Mendler-style (in
Sections~\ref{sec:initiality} and \ref{sec:explicit_flattening}) but
the formulation of the fusion law to come next is more natural in the
more general setting (no fusion law was needed in
\cite{DBLP:journals/tcs/MatthesU04} since no \emph{morphisms} of
heterogeneous substitution systems were considered there).

The next lemma shows a sufficient condition for two applications of the iterator
   $\giter{}{}{}{-}$ to be related:

\begin{lemma}[Fusion law]\label{fusion_law}
   Suppose the data as given in Lemma~\ref{lem:gen_mendler_it}.
   Additionally, let $L' : \C\to\D$ be a functor, $X' : \D_0$ be an object of $\D$,
 let
    \[
     \Psi' : \D(L' -, X') \to \D(L'(F -) , X')
    \]
     be a natural transformation with type analogous to that of $\Psi$, and
   let
    \[
      \Phi : \D(L -, X) \to \D(L' - , X')
    \]
    be a natural transformation.
     Then we have
     \[
         \Phi_{\mu F}\bigl(\giter{F}{L}{}{\Psi}\bigr) = \giter{F}{L'}{}{\Psi'}
     \]
     if
     \[
        \Phi_{F\mu F} \circ \Psi_{\mu F} = \Psi'_{\mu F} \circ \Phi_{\mu F} \enspace .
     \]

\end{lemma}
The name ``fusion law'' is wide-spread in functional programming for
means to eliminate the creation of some extra structure, here the
subsequent calculation of $\Phi_{\mu F}$ for the result
$\giter{F}{L}{}{\Psi}$ of the iteration over $\mu F$ is ``fused'' into
one single iteration over $\mu F$---the right-hand side of the
conclusion.

\smallskip
The version of this fusion law with $X$ and $X'$ the same object of
$\D$ and instantiated to the special situation of generalized folds
(see Section~\ref{sec:initiality}) has been found by Bird and Paterson
\cite{DBLP:journals/fac/BirdP99} (see right before their Theorem 1). While we
will only use the fusion law for generalized folds (in
Section~\ref{sec:initiality}), it is necessary to have the liberty in
choosing $X$ and $X'$ separately. The proof itself is a matter of
verifying that the left-hand side satisfies the defining equation
(embodied in the commuting diagram in Lemma~\ref{lem:gen_mendler_it})
of the right-hand side. This also settles existence of the right-hand
side, which is why we did not require a right adjoint for $L'$, which
would have allowed us to invoke Lemma~\ref{lem:gen_mendler_it} also
for $\Psi'$. (In our formalization, we did not implement this subtlety
but require a right adjoint for $L'$, in order to use the definition
of the $\giter{}{}{}{-}$ operator underlying the formalization of
Lemma~\ref{lem:gen_mendler_it}.)

\section{The category of heterogeneous substitution systems}\label{sec:hss}
 In \cite{DBLP:journals/tcs/MatthesU04}, implicitly there is a notion of signature.
 Here, we make this definition explicit and adapt it to the lack of strictness of our monoidal structures
 on endofunctors (see Definition \ref{def:monoidal_endofunctors}) -- recall that $U$ ``forgets'' the points of pointed functors:

\begin{definition}[Signature]\label{def:signature}
  Given a category $\C$, a \fat{signature} is a pair $(H,\theta)$ of an endofunctor $H$ on $[\C,\C]$ and a natural transformation
  $\th{}: (H {-}) \cdot U {\sim} \arr H ({-} \cdot U {\sim})$
  between functors
  $[\C,\C] \times \Ptd(\C) \arr [\C,\C]$
  such that
  \[ \th{X, \idwt{\Ptd(\C)}} = H(\lambda^{-1}_X) \comp \lambda_{HX} \]   
  and
  \[  \th{X, (Z' \cdot Z, e' \cdot e)} = H(\alpha^{-1}_{X,Z',Z}) \comp   \th{X \cdot Z', (Z,
    e)} \comp (\th{X, (Z', e')} \cdot Z) \comp \alpha_{HX, Z', Z} \enspace . \]   
\end{definition}

In practice, a signature is given by a family of \emph{arities}, each arity specifying the type of a term constructor.
The above definition of signature is modular in the sense that building a signature from arities corresponds to taking
an amalgamated sum. This is explained in detail in Section \ref{sec:explicit_flattening}, to which we refer
for an example of signature.

Note that while the definition of signature does not require the base category $\C$ to have coproducts,
this is a requirement for most signatures that we consider in practice, and in particular for the
example of Section \ref{sec:explicit_flattening}.
It also is a requirement for the definition of \enquote{models} of that signature, see Definition \ref{def:subst-systems}.

\begin{convention}
 From now on, we assume the category $\C$ to have (specified) coproducts.
 We denote by $\inl_{A,B} : A \to A + B$ and $\inr_{A,B} : B \to A + B$ the maps into the coproduct.
 We omit the subscripts of $\inl$ and $\inr$ when possible without ambiguity.
\end{convention}

\begin{remark}
  The notion of signature introduced in Definition \ref{def:signature} encompasses
  \enquote{polynomial} signatures like the ones described in \cite{fpt} and \cite{DBLP:conf/ppdp/MiculanS03}.
  In fact, it is strictly more general in that it also encompasses the arity of explicit
  flattening---the Example \ref{def:arity_flat} we discuss in detail in Section \ref{sec:explicit_flattening}---that is not captured by
  the other works mentioned above.
\end{remark}

For a given signature $(H,\theta)$, we are interested in $(\constfunctor{\Id{\C}}+H)$-algebras
$(T,\alpha)$.
For such an algebra, the natural transformation $\alpha : \Id{\C} + HT \to T$
decomposes
into two $[\C,\C]$-morphisms
$\et{} : \Id{C} \arr T$, $\ta{} : H T \arr T$ defined by
\begin{equation}
\begin{array}{c}
\et{} = \al{} \comp \inl_{\Id{\C},HT} \qquad\mbox{and}\qquad
\ta{} = \al{} \comp \inr_{\Id{\C},HT} \enspace .
\end{array}\label{hss_eta_tau}
\end{equation}
The pair $(T,\eta)$ is an object in the category of pointed functors (see Definition \ref{def:pointed_functor}).

Intuitively, in the case where $\C = \Set$,
the transformation $\et{}$ corresponds to viewing variables $x : X$ as \enquote{terms},
that is, as elements of $TX$
whereas $\ta{} : HT \to T$ represents the recursive constructors specified by $H$.

\begin{definition}[Def.~5 of \cite{DBLP:journals/tcs/MatthesU04}, Heterogeneous substitution system of a signature]\label{def:subst-systems}
We call $(T, \al{})$ a \emph{heterogeneous substitution system} (or \enquote{hss} for short) for
$(H, \th{})$, if, for every $\Ptd(\C)$-morphism $f : (Z,e) \arr (T,
\et{})$, there exists a unique $[\C,\C]$-morphism $h : T \cdot Z \arr
T$, denoted $\gst{(Z,e)}{f}$, satisfying
\[
\begin{array}{c}
\xymatrix@C2pc@R1pc{
**[l] Z + (H T) \cdot Z \ar[d]_{\idwt{Z} + \th{T,(Z,e)}} \ar[r]^-{\al{} \cdot Z}
  & T \cdot Z \ar[dd]^{h} \\
**[l] Z + H (T \cdot Z) \ar[d]_{\idwt{Z} + H h}
  & \\
**[l] Z + H T \ar[r]^-{\copair{f}{\ta{}}}
  & T
}
\hspace{2mm} \mbox{i.e.,} \hspace{2mm}
\xymatrix@C3pc@R1pc{
Z \ar[r]^{\et{} \cdot Z} \ar[ddr]_{f}
    & T \cdot Z \ar[dd]^{h}
        & (H T) \cdot Z  \ar[l]_-{\ta{} \cdot Z} \ar[d]^{\th{T, (Z,e)}} \\
    &
        & H (T \cdot Z) \ar[d]^{H h} \\
    & T
        & H T  \ar[l]_-{\ta{}}
}
\end{array}
\]
For a substitution system $(T,\alpha,\gst{}{-})$, we call $T$ its \fat{carrier}, thus extending
the convention of Definition \ref{def:algebra}.
\end{definition}
Notice that the quantification is implicitly also over all pointed endofunctors $(Z,e)$ on $\C$.

In the following, we sometimes omit the word \enquote{heterogeneous} when talking about heterogeneous substitution systems.

\begin{remark}\label{rem:bracket_prop}
Being equipped with a \enquote{bracket} operation $\gst{}{-}$ is a proposition on $(\constfunctor{\Id{\C}} + H)$-algebras.
\end{remark}
Notice that we call the operation a bracket operation although we write it with braces, to distinguish it from the bracket notation used for parallel substitution in the introduction.

\smallskip
The statement of the following lemma is mentioned, but not proven in \cite{DBLP:journals/tcs/MatthesU04}:

\begin{lemma}
  The operation $\gst{}{-}$ is a natural transformation \[\Ptd(- , (T,\eta)) \to [\C,\C](T \cdot U - , T) \enspace . \]
\end{lemma}

Note that the substitution operation given by the bracket is not categorical in the sense that it is not given by
a universal property. This is due to the fact that we prefer an operational point of view, where things actually compute,
over a categorical one.
Having substitution given as an operation rather than via a universal property is also crucial
for obtaining a monad, that is, for the main theorem of \cite[Thm.\ 10]{DBLP:journals/tcs/MatthesU04}.

\begin{definition}[Category of substitution systems]\label{def:cat_of_hss}
 Given $(H,\theta)$ as before, the category $\HSS(H,\theta)$ has, as objects, heterogeneous substitution systems as in \Cref{def:subst-systems}.
 A morphism of substitution systems is an algebra morphism that is compatible with
 the bracket $\gst{}{}$ on either side.
 In terms of $\eta$ and $\tau$ as defined in Equation \eqref{hss_eta_tau},
 a morphism from $(T,\eta,\tau,\gst{}{})$ to $(T', \eta', \tau', \gst{}{}')$ is a natural transformation $\beta : T \to T'$ such that
 the following diagrams commute:
 \[
  \begin{xy}
   \xymatrix{\Id{} \ar[r]^{\eta} \ar[rd]_{\eta'}  & T \ar[d]^{\beta} \\
                  & T'}
  \end{xy} \hspace{10pt}
  \begin{xy}
   \xymatrix{
        HT \ar[r]^{\tau} \ar[d]_{H\beta}& T\ar[d]^{\beta} \\
        HT' \ar[r]_{\tau'}& T'
   }
  \end{xy} \hspace{10pt}
  \begin{xy}
   \xymatrix{
       T \cdot Z \ar[r]^{\gst{}{f}} \ar[d]_{\beta \cdot Z} & T \ar[d]^{\beta} \\
       T'\cdot Z \ar[r]_-{\gst{}{\beta \circ f}'}& T'
   }
  \end{xy}
 \]
Here, the first and second diagram express the property of $\beta$ being an algebra morphism, and the third diagram expresses
compatibility of $\beta$ with substitution on either side.

Note that the composite $\beta \circ f$ in the last diagram is the composite in the category of \fat{pointed} endofunctors,
that is, the definition of that composite uses commutativity of the first diagram.
\end{definition}

\begin{remark}\label{rem:bracket_mor_prop}
Similarly to Remark~\ref{rem:bracket_prop}, being compatible with the brackets on either side is a proposition on algebra morphisms.
\end{remark}

We now study the category $\HSS(H,\th{})$ of substitution systems associated to a signature in
more detail, in particular with respect to the particular foundations we are working in.
The main objective of the rest of the section is Theorem~\ref{thm:hss_sat}: the category $\HSS(H,\th{})$
is univalent if the base category $\C$ is.

Remarks \ref{rem:bracket_prop} and \ref{rem:bracket_mor_prop} together
show that the category of $\HSS(H,\theta)$ can be obtained as a \emph{subcategory}
of the category of $(\constfunctor{\Id{}}+H)$-algebras in the following sense:

\begin{definition}
  A \fat{subcategory} of a category $\C$ is given by
  a predicate $P : \C_0 \to \prop$ and a family of predicates $P_{a,b} : P(a)\times P(b)\times \C(a,b) \to \prop$
  that is closed under identity and composition in the sense that
  \begin{itemize}
   \item for any $a : \C_0$ satisfying $P$, have a proof of $P_{a,a}(\identity(a))$ and
   \item for any $a,b,c:\C_0$ satisfying $P$, and for any $f : \C(a,b)$ and $g:\C(b,c)$, have
       a map $P_{a,b}(f) \to P_{b,c}(g) \to P_{a,c}(g\circ f)$.
  \end{itemize}
 We suppress the arguments of type $P(a)$ and $P(b)$ when discussing the predicate $P_{a,b}(f)$, since those arguments are unique.
\end{definition}

A subcategory of $\C$ is---better, gives rise to---a category $\C_P$;
objects are of the form $\sm{x:\C_0}P(x)$,
and morphisms $(f,p_f) : \C_P\bigl((a,p_a),(b,p_b)\bigr)$ are pairs of a morphism $f : \C(a,b)$ of $\C$
together with a proof $p : P_{a,b}(f)$.

Given a signature $(H,\th{})$, define a subcategory of the category of  $(\constfunctor{\Id{}}+H)$-algebras via the predicates of
Remarks \ref{rem:bracket_prop} and \ref{rem:bracket_mor_prop}.
The resulting category is clearly isomorphic to $\HSS(H,\theta)$ in the sense of \cite[Definition~6.9]{rezk_completion}.

Note that isomorphic categories are equal modulo propositional equality \cite[Definition~6.16]{rezk_completion},
and hence share all properties definable in type theory.
We thus give up the distinction between the category $\HSS(H,\theta)$ and the subcategory of $(\constfunctor{\Id{}}+H)$-algebras it is isomorphic to.

A subcategory is called \fat{replete}, when it is closed under isomorphism, that is,
when, for $f : \iso_\C(a, b)$ and $P(a)$, it follows that $P(b)$ and $P_{a,b}(f)$.

\begin{lemma}\label{lem:hss_replete}
 The category $\HSS(H,\theta)$ is a replete subcategory of the category of $(\constfunctor{\Id{}}+H)$-algebras.
\end{lemma}

\begin{proof}
  Given a substitution system $(T,\alpha,\gst{}{-})$, an algebra $(T',\alpha')$ and an algebra isomorphism
  $\beta : (T,\alpha) \to (T',\alpha')$, we define a bracket $\gst{}{-}'$ on $(T',\alpha')$ as follows:
  for a given pointed morphism $f : (Z,e) \to (T',\eta')$, we define $\gst{}{f}'$ as the composition
  \begin{equation*}
    \gst{}{f}' :=  \beta \circ \gst{}{\beta^{-1} \circ f} \circ \beta^{-1}\cdot Z \enspace : \enspace T' \ot T \ot  T\cdot Z \ot T'\cdot Z 
  \end{equation*}
  The morphism $\gst{}{f}'$ thus defined satisfies the equations of Definition~\ref{def:subst-systems},
  \begin{align*}
     f &=  \gst{}{f}' \circ \eta'\cdot Z \\
     \gst{}{f}'\circ  \tau'\cdot Z &= \tau' \circ H(\gst{}{f}') \circ \th{T',(Z,e)} \enspace ;
  \end{align*}
  the calculation is routine.
  Concerning the uniqueness of $\gst{}{f}'$, suppose $h$ such that
  these equations with $h$ in place of $\gst{}{f}'$ are satisfied.
  We have to show that $h = \beta \circ \gst{}{\beta^{-1} \circ f} \circ \beta^{-1}\cdot Z$.
  Equivalently, one can show that
  \begin{equation}
    \gst{}{\beta^{-1} \circ f} = \beta^{-1} \circ h \circ \beta\cdot Z \enspace , \label{eq:uniqueness}
  \end{equation}
  which follows from the uniqueness of $\gst{}{-}$: it suffices to show that the right-hand side of
  \eqref{eq:uniqueness} satisfies the equations involving $\eta$ and $\tau$.
  We thus have equipped $(T',\alpha')$ with a (necessarily unique) substitution operation.

  The fact that $\beta$ is compatible with $\gst{}{-}$ and $\gst{}{-}'$, and hence in the subcategory, is a routine calculation.
\end{proof}

\begin{theorem}\label{thm:hss_sat}
 The category $\HSS(H,\theta)$ is univalent if \C is.
\end{theorem}

\begin{proof}
 Combine Lemmas~\ref{lem:alg_sat}, \ref{lem:sub_sat}, \ref{lem:functor_cat_sat} and \ref{lem:hss_replete}.
 More precisely, if $\C$ is univalent, so is $[\C,\C]$, and thus also the category
 of $(\constfunctor{\Id{}}+H)$-algebras on $[\C,\C]$.
 Finally, the category $\HSS(H,\th{})$ is univalent as a replete subcategory of that of $(\constfunctor{\Id{}}+H)$-algebras.
\end{proof}

The following lemmas state closure properties of the property of being univalent:

\begin{lemma}\label{lem:alg_sat}
 The category of algebras of a functor $F : \C \to \C$ is univalent if $\C$ is.
\end{lemma}

\begin{proof}
  This lemma is proved in the file \texttt{CategoryTheory/FunctorAlgebras.v} of the \UniMath library.
\end{proof}

The next lemma is originally due to Hofmann and Streicher \cite{hs:gpd-typethy};
and is also proved in Thm.~4.5 of \cite{rezk_completion}:

\begin{lemma}\label{lem:functor_cat_sat}
 The category of functors $[\C,\D]$ is univalent if the target category $\D$ is.
\end{lemma}

The category of hss contains all the isomorphisms of the category of $(\constfunctor{\Id{}}+H)$-algebras,
for which source and target are substitution systems.
This is sufficient to inherit univalence from the category of algebras:

\begin{lemma}\label{lem:sub_sat}
 Let $\C$ be a univalent category and let $P : \C_0 \to \prop$ and $P_{a,b} : \C(a,b) \to \prop$ define a subcategory $\C_P$ of $\C$.
 Then $\C_P$ is univalent if, for any objects $(a,p_a)$ and $(b,p_b)$ of $\C_P$, and for any \emph{iso}morphism $f : \iso_{\C}(a,b)$ from $a$ to $b$,
  we have $P_{a,b}(f)$.

 In particular, replete subcategories of univalent categories are univalent.
\end{lemma}

\begin{proof}
 For $(a,p_a)$ and $(b,p_b)$ objects of $\C_P$, we have
   \[ (a,p_a) =_{\C_P} (b,p_b) \enspace \simeq \enspace a =_{\C} b \enspace \simeq \enspace \iso_{\C}(a,b) \enspace \simeq \enspace \iso_{\C_P}((a,p_a),(b,p_b)) \]
 and this equivalence, from left to right, is equal to $\idtoiso$.
\end{proof}

This concludes our study of the category of substitution systems associated to a signature.

\section{From substitution systems to monads}\label{sec:monads}

One of the most important results of Matthes and Uustalu's work \cite{DBLP:journals/tcs/MatthesU04}
is the construction of a monad from any substitution system:

 \begin{theorem}[\cite{DBLP:journals/tcs/MatthesU04}, Thm.~10] \label{thm:mon2}
  If an $(\constfunctor{\Id{C}} + H )$-algebra $(T, \alpha)$ forms a heterogeneous
  substitution system for $(H,\th{})$ for some $\th{}$, then $(T, \et{}, \gst{}{\identity_{(T,\et{})}})$ is
  a monad.
 \end{theorem}

See Section~\ref{sec:formalization} for some comments on technical challenges we had to overcome for the formalization of its proof.

It is natural to ask whether this map extends to \emph{morphisms}, and indeed it does:

\begin{theorem}\label{thm:functor_hss_mon}
  The map from heterogeneous substitution systems to monads defined in \cite[Thm.\ 10]{DBLP:journals/tcs/MatthesU04}
  is the object map of a functor $\HSS(H,\theta) \to \Mon(\C)$.
\end{theorem}
\begin{proof}
  Given any morphism $\beta : (T,\eta,\tau,\gst{}{}) \to (T', \eta', \tau', \gst{}{}')$ of hss,
  the underlying natural transformation $\beta : T \to T'$
  needs to be proven compatible with the multiplications $\mu^T := \gst{}{\identity_{(T,\et{})}}$
  and $\mu^{T'}$
  of the monadic structures on $T$ and $T'$ defined in \cite[Thm.\ 10]{DBLP:journals/tcs/MatthesU04}.
  This is an easy consequence of the compatibility of $\beta$ with $\gst{}{}$ and $\gst{}{}'$.
\end{proof}

The functor from substitution systems to monads is faithful, but not full.
Intuitively, the lack of fullness stems from the fact that the axioms of a monad morphism
do not specify compatibility of the mapping with the \enquote{inner nodes} of an expression, but only
at the leaves, that is, in the case of a variable.

\begin{lemma}
 The functor of \Cref{thm:functor_hss_mon} is faithful.
\end{lemma}
\begin{proof}
  Two parallel monad morphisms are equal if their underlying natural transformations are,
  and the analogous statement is true for morphisms of substitution systems.
\end{proof}

\begin{remark}
 The functor of \Cref{thm:functor_hss_mon} is not full.
 For instance, choose $\C=\Set$, and take a signature with two copies $\app{}$ and $\app{}'$ (of the same arity) of an
 \enquote{application} constructor, see Definition \ref{def:arity_app} in Section~\ref{sec:explicit_flattening}.
 Take the initial substitution system associated to that signature
 (as constructed via Theorems~\ref{lem:mu15} and \ref{thm:lift_initial_is_initial} in Section~\ref{sec:initiality}), and define
 an endomorphism on it that maps $\app{}$ to $\app{}'$ recursively, and is the identity on the other constructors.
 This yields
 a monad morphism, but not a morphism of substitution systems;
 indeed, the second diagram of Def.~\ref{def:cat_of_hss} does not commute---any endomorphism on that substitution system
 must be the identity morphism.
\end{remark}

\section{Lifting initiality through a fusion law}\label{sec:initiality}

The starting point of this section is a result from
\cite{DBLP:journals/tcs/MatthesU04}, which gives one way to define
substitution systems and which comes from a very specific instance of
Lemma~\ref{lem:gen_mendler_it}. As a first instantiation step, take in
that lemma $[\C,\C]$ for $\C$ and $\D$ and the \emph{reduction
  functor} ${-} \cdot Z$ for $L$, for any endofunctor $Z$ of
$\C$. This is the general situation of the ``gfolds'' of Bird and
Paterson \cite{DBLP:journals/fac/BirdP99}, and (the carriers of) the
corresponding initial $F$-algebras are called ``nested datatypes''
\cite{birdmeertens}. As Bird and Paterson recall, the assumption of
having a right adjoint to the reduction functor means that right Kan
extensions along those $Z$ exist. In the context of functional
programming with impredicative polymorphism, these right Kan
extensions even exist in a computational way (although the full
categorical properties of Kan extensions are not reflected
computationally) \cite{grossestcspaper}. We will not further develop
the categorical semantics of those programming languages. The previous
remarks should make it plausible that the following theorem rests on
``reasonable'' technical conditions. If program verification is aimed
at in an intensional setting, replacements for the categorical notions
have to be found, and yet different schemes of generalized iteration
have to be studied in order to combine expressivity, termination
guarantees and program verification in the same framework
\cite{RalphSCP} (using Coq very differently from the \UniMath
approach).

\begin{theorem}[\cite{DBLP:journals/tcs/MatthesU04}, Thm.~15]\label{lem:mu15}
  Let $(H,\theta)$ be a signature.
  If $[\C,\C]$ has an initial $(\constfunctor{\Id{\C}} + H)$-algebra and a right
  adjoint for the functor ${-} \cdot Z : [\C,\C] \arr [\C,\C]$ exists for
  every $\Ptd(\C)$-object $(Z,e)$, then $(T, \al{})$ defined by
\[
(T, \al{}) = (\mu (\constfunctor{\Id{\C}} + H), \In{\constfunctor{\Id{\C}} + H})
\]
is a heterogeneous substitution system for $(H,\th{})$.
\end{theorem}

The proof of this theorem is by identifying, for a given $f : (Z,e)
\arr (T,\et{})$, the morphism $\gst{(Z,e)}{f}$ as an instance of
Lemma~\ref{lem:gen_mendler_it}, both for the existence and uniqueness
property. The obvious part of the instantiation is the choice of
parameters mentioned above, and by setting $F:=\constfunctor{\Id{\C}}
+ H$. The essential ingredient for getting a morphism $\gst{(Z,e)}{f}$
of type $\mu F\cdot Z\arr T$ (here, $T$ is even $\mu F$) is a natural
transformation $\Psi_f$ whose typing could sloppily be written as
$$\Psi_f::\forall X:[\C,\C].\,(X\cdot Z\to T)\to (FX\cdot Z\to T)\enspace.$$
The type of $\Psi_f$ suggests
the following problem-solving method: The original problem is that of
finding a morphism of type $\mu F\cdot Z\arr T$. We abstract away from
$\mu F$ and replace it by an arbitrary endofunctor $X:[\C,\C]$. For
this arbitrary $X$, we have to extend a purported solution for
parameter $X$, hence of type $X\cdot Z\to T$, to a solution for
parameter $FX$, hence of type $FX\cdot Z\to T$. Of course, this has to
be done naturally in $X$, as required in
Lemma~\ref{lem:gen_mendler_it}. So, passing naturally from $X$ to $FX$
as parameter, the lemma even yields a (unique) solution for the least
fixed-point of $F$ as parameter. The continuity properties behind this
method already for (co-)inductive types have been deeply explored by
Abel \cite{DBLP:journals/ita/Abel04} and extended to nested dataypes
later \cite{AbelThesis}.

This is the essence of schemes in Mendler's style
\cite{DBLP:journals/apal/Mendler91}: passing from a solution in
parameter $X$ to a solution in parameter $FX$ \emph{uniformly} (in
Mendler's original work, this was plainly universal quantification
over a type variable $X$, in the categorical setting, this is achieved
by naturality), one is guaranteed a solution in parameter $\mu
F$. Lemma~\ref{lem:gen_mendler_it} is an instance of that idea, hence
the name generalized iteration in Mendler-style.

Mendler-style gives great liberty: were are free in choosing $\Psi_f$
of the required type (implicitly asking for naturality), but there is
little guidance in finding the right one for our purpose. Guidance
would, e.\,g., come from asking for an algebra structure on the target
endomorphism $T$.  Therefore, we instantiate the lemma further to
obtain what is called ``a special case of generalized iteration'' by
Matthes and Uustalu \cite{DBLP:journals/tcs/MatthesU04}.\footnote{The
  instantiation with ${-} \cdot Z$ for $L$ can also be formulated in a
  less homogeneous setting where not only endofunctor categories
  intervene \cite[Section 2.3]{DBLP:journals/tcs/MatthesU04}.}
It consists in requiring an endofunctor $F'$ on $[\C,\C]$, a natural transformation $\theta':(F-)\cdot Z\arr F'(-\cdot Z)$ and an $F'$-algebra $\phi:F'T\arr T$ on $T$, and in putting them together to obtain
$$\Psi_f(X)(h:X\cdot Z\arr T):=\phi\circ F'h\circ \theta'_X:FX\cdot Z\arr T\enspace.$$
Its use in our present situation is then with
$F':=\constfunctor{Z}+H$, $\theta'_X:=\idwt{Z}+\theta_{X,(Z,e)}$ and
$\phi:=[f,\tau]$, using the datum $\theta$ of the signature and the
$H$-algebra $\tau$ that is generically derived from $\alpha$ (see
before Definition~\ref{def:subst-systems}).

\smallskip We remark that all of this is not optimal from a
progammer's point of view (the question is then not only of soundness
but of efficiency of the traversals through the data structures) and
that there is the more refined notion of ``generalized Mendler
iteration'' \cite{grossestcspaper} (called $\mathsf{GMIt^\omega}$) as
an efficient way out. The crucial idea is to generalize the problem
further than finding a solution of $X\cdot Z\arr T$ for parameter
$X=\mu F$. An $h:X\cdot Z\arr T$ consists of morphisms $h_A:X(ZA)\arr
TA$ for every $A:\C_0$, and generalized Mendler iteration asks even
for operations $h_f:XB\arr TA$ for any $B:\C_0$ and $f:B\arr
ZA$. Taking for $f$ the identity morphism on $ZA$, one gets the
desired components of the solution in the end. The gain in efficiency
comes from the combination of a fold and a map in this
scheme---enforced just by these types in the polymorphic formulation
of \cite{grossestcspaper}.

Also for generalized Mendler iteration, there is a
formulation in more conventional terms of algebras, called ``generalized
refined conventional iteration'' \cite{grossestcspaper}, which captures
in particular the efficient folds of Martin, Gibbons and Bayley
\cite{gibbons:efolds}. For generalized Mendler iteration, there is
also a means of verification in usual intensional \Coq, using category
theory only as a motivation and not as the mathematical framework
\cite{RalphSCP}.  
\medskip

We augment the previous theorem by showing that the constructed substitution system is initial:

\begin{theorem}\label{thm:lift_initial_is_initial}
 The substitution system $(T,\alpha,\gst{}{})$ constructed in Lemma~\ref{lem:mu15} is initial in $\HSS(H,\theta)$.
\end{theorem}

In order to prove Theorem~\ref{thm:lift_initial_is_initial}, it suffices to show that,
for any given substitution system $(T',\alpha', \gst{}{}')$,
the initial morphism \fat{of algebras}
\[! : (T,\alpha) \to (T',\alpha')\]
is compatible with the  operations $\gst{}{}$ (defined in the proof of Lemma~\ref{lem:mu15}) and $\gst{}{}'$.
That is, we need to show that, for any $f : (Z,e)\arr (T,\et{})$,
\begin{equation}
    !\circ \gst{}{f} = \gst{}{!\circ f}'\circ (! \cdot Z) \enspace . \label{eq:init_compat}
\end{equation}
Using the fusion law (Lemma~\ref{fusion_law}), we show that both sides
of \eqref{eq:init_compat} are equal to the application of an
iterator. More precisely, we use the fusion law for the left-hand
side, knowing the explicit definition of $\gst{}{f}$ as an iterator,
described above, to establish equality with $\giter{F}{{-} \cdot
  Z}{}{\Psi_f}$, where we define
$$\Psi_f(X)(h:X\cdot Z\arr T'):=[!\circ f,\tau'\circ Hh\circ \theta_{X,(Z,e)}]:FX\cdot Z\arr T'\enspace.$$
Once the premisses of the fusion law established, we can show equality
with the right-hand side of \eqref{eq:init_compat} by verifying that
the defining equations of $\giter{F}{{-} \cdot Z}{}{\Psi_f}$ are
fulfilled by the right-hand side.

\section{A worked example: flattening of explicit substitution}\label{sec:explicit_flattening}
In practice, a signature is often a family of arities, each arity
specifying the type of one term constructor.  A typical example is a
typeful version of de Bruijn indices for pure (untyped)
$\lambda$-calculus, where, intuitively, the equation
$$TA=A+TA\times TA+T(1+A)$$
has to be solved, giving in $TA$ the set of $\lambda$-terms having free
variables among $A$ (cf.~the introduction), where the last summand
represents $\lambda$-abstraction that abstracts the variable
corresponding to the extra element of $1+A$. This example is developed
in \cite{DBLP:journals/tcs/MatthesU04} but originates in
\cite{alticsl99,birdpatersonJFP}.

In our formalism (that of \cite{DBLP:journals/tcs/MatthesU04}), we do
not need to distinguish between arities and signatures.  Intuitively,
an arity is a signature that is not obtained as a proper sum of two
other signatures.  In particular, a single arity constitutes a
signature, and we can \enquote{glue} signatures together to obtain a
new signature:

\begin{lemma}[Sum of signatures]\label{lem:sum_of_signatures}
  Let $(H,\theta)$ and $(H',\theta')$ be two signatures. Then $(H+H',\theta+\theta')$ is a signature.
\end{lemma}

This lemma is important for our main example: indeed, we consider two
signatures, where one is obtained from the other by extending the
language (better: its signature) by one additional term constructor
(better: arity).

To this end, we need the base category $\C$ to come equipped with some
extra structure: for the remainder of this section, we assume $\C$ to
have (specified) products, coproducts and a terminal object.  An
example of such a category is the (univalent) category $\Set$ of sets
(see Section~\ref{sec:univalent_mathematics}), which has all limits
and colimits.

We continue the case study in \cite{DBLP:journals/tcs/MatthesU04} on
$\lambda$-calculus without and with a form of explicit
substitution---``explicit flattening''.  In order to do so, we first
present the functors $H$ and natural transformations $\theta$
corresponding to the arities of application, abstraction, and explicit
flattening, respectively:

\begin{definition}[application]\label{def:arity_app}
  The signature of application is given by pointwise product, inherited from the
  base category $\C$:
  \[H^{\App}(T) := T \times T \enspace . \]
  The natural transformation $\theta^{\App}$ is given pointwise by the identity,
    \[\th{X,(Z,e)}^{\App} : (X \times X) \cdot Z \to (X\cdot Z)\times (X\cdot Z) \enspace . \]
\end{definition}
The fact that the identity suffices here corresponds to the triviality
of first-order operations in substitution (which is plainly homomorphic on those operations).

\begin{definition}[abstraction]\label{def:arity_abs}
  Abstraction in our context is defined by precomposition with
  a coproduct, corresponding to \enquote{context extension}:
 \[H^{\Abs}(T) := T \cdot \option \enspace ,  \]
 where $\option(X) := 1 + X$ represents the context $X$ extended by one distinguished element $\inl_{1,X}(\star)$.
 The \enquote{strength} $\th{}$ is defined as
  \[\th{X,(Z,e)}^{\Abs}(A) := X\copair{e_{1+A} \comp
     \inl_{1,A}}{Z \inr_{1,A}}:X(1+ZA)\arr X(Z(1+A)) \enspace . \]
\end{definition}
The defined strength embodies the usual lifting needed for
substitution in de Bruijn representations of $\lambda$-abstraction.

\begin{definition}[explicit flattening]\label{def:arity_flat}
  The flattening signature is defined by selfcomposition,
  \[H^{\Flatten}(T) := T \cdot T \enspace , \]
  and the corresponding strength requires the unit $e$ of the pointed endofunctor $(Z,e)$
  to be inserted in the right place:
  \[\th{X,(Z,e)}^{\Flatten}:=X\cdot e\cdot X\cdot Z:X\cdot X\cdot Z\arr X\cdot Z\cdot X\cdot Z \enspace . \]
\end{definition}
Note that the flattening signature cannot be dealt with in a framework
with a fixed enumeration of variable names and shows, already on the
syntactic side, the most simple case of ``true nesting'' in nested
datatypes (see, e.\,g., \cite{grossestcspaper}). Notice that the
highly parameterized type already suggests the right definition. For
its mainly used instance $\th{T,(T,\eta)}^{\Flatten}$, with $T$ and
$\eta$ components of the obtained substitution system, its type
$T^3\arr T^4$ hardly suggests a canonical definition.

\smallskip
These signatures are now combined, as per Lemma~\ref{lem:sum_of_signatures}, to obtain
the signatures we are mainly interested in:

\begin{definition}[$\lambda$-calculus]\label{def:signature_lc}
  The signature $\lc$ is obtained as the sum of the signatures of
  Defs.~\ref{def:arity_app} and \ref{def:arity_abs}.
\end{definition}

\begin{definition}[$\lambda$-calculus with explicit flattening]\label{def:signature_lcflat}
  The signature $\lcmu$ is obtained as the sum of the signatures of
  Defs.~\ref{def:signature_lc} and \ref{def:arity_flat}.
\end{definition}

For the purpose of this example, we assume the signatures $\lc$ and
$\lcmu$ to have initial substitution systems.  By Lemma \ref{lem:mu15}
we get those if we assume that their underlying initial algebras
exist.  (For a remark on the construction of initial algebras, see
Section \ref{sec:conclusions}.)  We denote the initial substitution
systems by $(\Lam,\alpha,\gst{}{})$ and
$(\Lammu,\alpha^\mu,\gst{}{}^{\mu})$, respectively. Intuitively, they
solve the equation in $T$ given in the first paragraph of this
section, and the following in $T'$, respectively:
$$T'A=A+T'A\times T'A+T'(\option\,A)+T'(T'A)\enspace.$$

Why is $\Lammu$ supposed to represent $\lambda$-calculus with explicit
flattening?  Coming back to parallel substitution on $T$ ($=\Lam$), as
mentioned in the introduction, we may study the substitution rule
$f:=\lambda x^{TB}.x$ of type $TB\arr TB$. Then, $\mu_B:=[f]:T(TB)\to
TB$ can be interpreted as doing the following: in a term whose free
variables have as names terms over $B$, those names are replaced by
themselves, but now seen as terms that are ``integrated'' into the
result term. In other words, $\mu_B$ removes the ``cross section''
between the trunk of the term and the term-like variable
leaves. Invoking Theorem~\ref{thm:mon2} for $(\Lam,\alpha,\gst{}{})$,
one obtains $\mu:=\gst{}{\identity_{(\Lam,\et{})}}:\Lam\cdot\Lam\arr\Lam$ as monad
multiplication on the monad of $\lambda$-terms, and the
above-mentioned parallel substitution can then be derived generically,
so as to obtain its components $\mu_B$ with the described
behaviour. In other words, the generic notion of monad multiplication
appears to have the behaviour of ``flattening'' a nested term
structure of type $T(TB)$ into one of type $B$ (for every $B$).  Now,
$\Lammu$ even has a term constructor, corresponding to the injection
of the last summand of the above equation into the left-hand side, and
so, the constructor is of type $\Lammu\cdot\Lammu\arr\Lammu$, which is
of the same type as the monad multiplication that is obtained by
invoking Theorem~\ref{thm:mon2} for
$(\Lammu,\alpha^\mu,\gst{}{}^{\mu})$. As a constructor, this operation
does \emph{not} denote the result of the flattening (here, even for
the extended syntax), but is a formal syntactic element and is thus
termed an ``explicit flattening''. Already in
\cite{DBLP:journals/tcs/MatthesU04}, it was shown that those explicit
flattenings can be resolved by evaluating any term with explicit
flattenings (from $\Lammu A$ for some $A$) into a term without
explicit flattenings (in $\Lam A$). We continue this case study by
using our extra categorical structure on substitution systems.

In the following, our goal is to construct a morphism of substitution systems from $\Lammu$ to $\Lam$.
This is not quite precise and needs refinement, since a priori, those two substitution systems are not in the same category.
More precisely, we are going to build a substitution system for the signature $\lcmu$, the underlying
carrier of which is the carrier $\Lam$.
To this end, we need to construct two ingredients:
firstly, we need a natural transformation $\mu^{\Lam} : H^{\Flatten}(\Lam) \to \Lam$ in order to obtain
a structure of $\constfunctor{\Id{}} + \lcmu$-algebra on $\Lam$.
Secondly, we equip this $\constfunctor{\Id{}} + \lcmu$-algebra with a bracket operation---which, of course, must be shown
compatible with the $\constfunctor{\Id{}} + \lcmu$-algebra structure in the sense of the
diagram of Definition \ref{def:subst-systems}.

Once this is done, we obtain, by initiality, a morphism of hss from the initial hss of $\lcmu$ to the newly constructed one,
the underlying algebra morphism of which is a morphism from $\Lammu$ to $\Lam$ that \enquote{does the right thing}:
mapping explicit substitution to substitution.

\begin{definition}[representation of flattening on $\Lam$]
  Let $\mu^{\Lam} : H^{\Flatten}(\Lam) \to \Lam$ be given by
  \[ \mu^{\Lam} := \gst{}{\identity_{\Lam}} : \Lam \cdot \Lam \to \Lam  \enspace . \]
\end{definition}

\begin{lemma}[substitution system of $\lcmu$ on $\Lam$]\label{hss_extend}

  The pair $(\Lam, \copair{\alpha}{\mu^{\Lam}})$ is an $\constfunctor{\Id{}} + \lcmu$-algebra.
  (Here, we have implicitly used associativity of the coproduct.)

  We define a bracket operation $\gst{}{}^{\Flatten}$ on this algebra by setting, for $(Z,e)$ and $f : (Z,e) \to (\Lam,\eta)$,
  \[
	\gst{}{f}^{\Flatten} := \gst{}{f} \enspace .
  \]
  This assignment
  yields a bracket operation on that algebra, and hence a substitution system
  $(\Lam,\copair{\alpha}{\mu^{\Lam}},\gst{}{}^{\Flatten})$ for the signature $\lcmu$.

\end{lemma}

\begin{proof}
 We need to show that $\gst{}{-}^{\Flatten}$ satisfies the equations of a bracket operation, see Definition \ref{def:subst-systems}.
 The diagrams can be checked for any \enquote{arity} individually, and
 for $\eta$, $\App$ and $\Abs$, the equations to check are exactly those satisfied by $\Lam$ as a substitution system for
 the signature $\lc$.
 The only non-trivial equation to check states that $\gst{}{-}^{\Flatten}$ is compatible with $\mu^{\Lam}$;
 we have to check that
 \[ \gst{}{f}^{\Flatten} \circ \mu^{\Lam}\cdot Z
    =
    \mu^{\Lam}\circ \Lam(\gst{}{f}^{\Flatten}) \circ \gst{}{f}^{\Flatten}\cdot \Lam\cdot Z \circ \Lam\cdot e\cdot \Lam\cdot Z\]
 We omit the details of this calculation here, and refer instead to
 the formal proof.
\end{proof}

 We thus have two objects in the category $\HSS(\lcmu)$, an initial object with underlying carrier $\Lammu$,
 and the object constructed in Lemma~\ref{hss_extend}, with underlying carrier $\Lam$.
 By initiality, we obtain a unique morphism of hss in this category.

\begin{definition}
  We call $\EVAL : \Lammu \to \Lam$ the morphism of substitution systems
  obtained by initiality.
  This map sends application and abstraction to themselves, respectively, and it sends
  the explicit flattening operator to its \enquote{evaluation}, that is, to a \enquote{flattened} term.
\end{definition}

  This morphism of hss gives rise, via functoriality of the monad construction (Theorem~\ref{thm:functor_hss_mon}),
  to a monad morphism; it is this morphism that is studied in Example 16 of \cite{DBLP:journals/tcs/MatthesU04}.
  Here, we have shown how that monad morphism arises from a morphism of substitution systems.

\section{About the formalization}\label{sec:formalization}

Most of the results presented in this article have been formalized, based on the
\UniMath library \cite{unimath}.
More precisely, all results except for
Theorem \ref{thm:hss_sat} and Lemmas \ref{lem:sub_sat} and \ref{lem:hss_replete} are proved in our formalization.

Our formalization started out as an independent repository, but has since been integrated
into \UniMath, as a package (subdirectory) called \texttt{SubstitutionSystems}.
The formalization can be inspected by cloning the \texttt{UniMath} repository on
Github, \url{https://github.com/UniMath/UniMath}, following the installation procedure
described there.

The \texttt{UniMath} library being under active development,
the organization of the packages is going to change: some code will be moved to
other, more fundamental, packages.
For the purpose of inspection of
the package \texttt{SubstitutionSystems} as described here,
it is hence convenient to stick with
a particular commit of the git repository, e.g., \texttt{commit 1ead81a}. 
The sections of this article roughly correspond to files in the formalization:
\begin{description}
 \item [\texttt{GenMendlerIteration.v}] corresponds to Section~\ref{sec:mendler};
 \item [\texttt{SubstitutionSystems.v}] corresponds to Section~\ref{sec:hss};
 \item [\texttt{MonadsFromSubstitutionSystems}] corresponds to Section~\ref{sec:monads};
 \item [\texttt{LiftingInitial.v}] corresponds to Section~\ref{sec:initiality}.
\end{description}

The code corresponding to Section~\ref{sec:explicit_flattening} is spread over several files:
\begin{description}
 \item [\texttt{SumOfSignatures.v}] corresponds to Lemma~\ref{lem:sum_of_signatures};
 \item [\texttt{LamSignature.v}] corresponds to Definitions~\ref{def:arity_app}, \ref{def:arity_abs}, \ref{def:arity_flat};
 \item [\texttt{Lam.v}] corresponds to the rest of Section~\ref{sec:explicit_flattening}.
\end{description}

To account for the evolution that is going to happen in the \UniMath library,
we provide an \enquote{interface} file \\
{\centering
\texttt{UniMath/SubstitutionSystems/SubstitutionSystems\textunderscore Summary.v}
}\\
containing pointers to the most important formalized theorems.

\subsection{Statistics}

Our library consists of a bit more than 4400 loc, plus 600 lines of comments%
\footnote{Note that the organization of the files is going to change over time, due to
reorganization of the library. In particular, contents may get moved to other parts of \UniMath
in the future.}.
Details are given in Table~\ref{table:hss_form}---numbers are taken from \texttt{commit 1ead81a}.
For comparison, for the same commit, the whole of \UniMath, including our library,
consists of about 37000 lines of code:
\begin{verbatim}
     spec    proof comments
    15053    22389     3987 total
\end{verbatim}

\begin{table}[]
\caption{Lines of code of the library \texttt{SubstitutionSystems}}\label{table:hss_form}
\begin{verbatim}
     spec    proof comments
       32       59       10 AdjunctionHomTypesWeq.v
       90      165      102 Auxiliary.v
       28       14        8 EndofunctorsMonoidal.v
       70      124       27 FunctorsPointwiseCoproduct.v
       70      113        7 FunctorsPointwiseProduct.v
       91      116       30 GenMendlerIteration.v
       28       21        7 HorizontalComposition.v
       79      407       72 LamSignature.v
      106      249       57 Lam.v
      236      518       61 LiftingInitial.v
      123      423       76 MonadsFromSubstitutionSystems.v
       26        0       12 Notation.v
       15        4        9 PointedFunctorsComposition.v
       36       61       11 PointedFunctors.v
       42       81       11 ProductPrecategory.v
       22        0       10 RightKanExtension.v
       82      211       40 Signatures.v
      155      326       53 SubstitutionSystems.v
       69      170       13 SumOfSignatures.v
     1400     3062      616 total
\end{verbatim}

\end{table}

\subsection{About performance: transparency vs.\ opacity}

One important aspect of computer proof assistants that are based on type theory is \fat{computation}.
Computation enables us to obtain some equalities for free.
For instance, in our formalization of (co)products in a functor category $[\C,\D]$ from (co)products in the target category $\D$,
the (co)product of two functors $F$ and $G$ \fat{computes} pointwise to the (co)product of the images, that is,
for instance $(F \oplus_{[\C,\D]} G)(c) \equiv Fc \oplus_{\D} Gc$.
Here, the notation $\equiv$ denotes definitional equality a.k.a.\ computation.
This is only true for a specific construction of (co)products in functor categories, of course;
in general, one can only expect $(F \oplus_{[\C,\D]} G)(c) \enspace \simeq_{\D} \enspace Fc \oplus_{\D} Gc$.
However, in order to keep the complexity of our proofs manageable for us, having
definitional equality instead of isomorphism was crucial.
We hence had to keep many category-theoretic constructions, such as (co)products in functor categories, \fat{transparent}.
Technically, this amounts to closing a proof using \texttt{Defined.} instead of \texttt{Qed.} in the \Coq proof assistant.

This lack of opacification, however, results in terms getting very large, making type checking more costly for the
machine. The transparency vs.~opacity issue can hence be restated as an issue of human vs.~machine friendliness.

Our approach to this issue was to opacify all the terms that we could afford opacifying,
either by moving them into lemmas by themselves, closing with \texttt{Qed.},
or by enclosing the corresponding sequence of tactics producing that term into an
\texttt{abstract (\ldots)} block.
The inconvenience of the latter method is that the block enclosed by \texttt{abstract} must be
\fat{one} tactic (composed using the chaining semicolon),
not a sequence of tactics. This method is hence only feasible for small subproofs.

Our library is quite slow to compile, due to the rather large proof terms arising when working with
rank 2 functors: some \texttt{Qed.} take very long to check.
A significant speedup was obtained in the file \texttt{MonadsFromSubstitutionSystems.v} by setting
the option \texttt{Unset Kernel Term Sharing.}, the workings of which are unknown to us.
However, this option proved useless or even increased compile time in other files, and is hence only
used in that one file.
It is unclear to us why this option is beneficial in that file and only there,
and whether there is a guiding principle saying when this option is useful.

In our library, there is a slight duplication of code:
the \UniMath library 
contains a proof that colimits lift to functor categories from the target category,
formalized by Ahrens and Mörtberg \cite{ahrens_mortberg}.
This result could in principle be applied to lift coproducts and products, both of which are formalized as specific colimits.
However, it turned out that this approach made typechecking unfeasibly slow:
indeed, the first files making use of coproducts in functor categories would stop compiling when that construction of
coproducts in functor categories was plugged in.
Instead, we provide a manual lifting of (co)products into functor categories in the files \texttt{FunctorsPointwiseProduct.v} and
\texttt{FunctorsPointwiseCoproduct.v}, with which typechecking is reasonably fast.
The latter construction applies similar principles of opacification as the general lifting of colimits; it is hence
unclear to us why the latter does perform so much better than the former.
We hope to clarify this issue in future work \cite{ahrens_mortberg}.

\section{Conclusions}\label{sec:conclusions}
 We presented, in a univalent foundation, some new results about the heterogeneous substitution
 systems introduced by Matthes and Uustalu \cite{DBLP:journals/tcs/MatthesU04},
 and showed how to obtain initial substitution systems (such as lambda calculi)
 from initial algebras using generalized iteration in Mendler-style.

  We have not studied the construction of initial algebras in univalent foundations;
  this is the subject of a forthcoming work by Ahrens and Mörtberg \cite{ahrens_mortberg}.

  Thanks to Paige North for discussion of the subject matter, and to
  Anders Mörtberg for providing feedback to a draft of this article.
  Thanks to the rest of the \UniMath team, for providing a sound base for formalization,
  and, specifically, to Dan Grayson and Anders Mörtberg for helping maintain the code described in this article.

\bibliographystyle{plain}
\bibliography{literature}



\end{document}